%% file: main.tex
\documentclass[10pt,onecolumn]{IEEEtran}
\usepackage{graphicx,amsmath,amsthm,amssymb,amstext}
\usepackage{color}
\usepackage[normalem]{ulem}
\newtheorem{theorem}{Theorem}
\newtheorem{definition}{Definition}
\newtheorem{proposition}{Proposition}
\newtheorem{corollary}{Corollary}

\newtheorem{lemma}{Lemma}
\newtheorem{remark}{Remark}

\begin{document}
\title{Privacy-Utility Tradeoffs under Constrained Data Release Mechanisms}

\author{Ye Wang~\IEEEmembership{Member,~IEEE}, Yuksel Ozan Basciftci~\IEEEmembership{Member,~IEEE}, and Prakash Ishwar~\IEEEmembership{Senior Member,~IEEE}
\thanks{Y.~Wang is with Mitsubishi Electric Research Laboratories (MERL), Cambridge, MA 02139, email:~{\tt yewang@merl.com}. Y.~O.~Basciftci is with Qualcomm, Boxborough, MA 01719, email:~{\tt yukselb@qti.qualcomm.com}. P.~Ishwar is with Boston University, Boston, MA 02215, email:~{\tt pi@bu.edu}.}
\thanks{Y.~O.~Basciftci performed this work during an internship at MERL.}
\thanks{An earlier version of part of this work appeared in~\cite{ITA2016}.}}

\maketitle

\begin{abstract}
Privacy-preserving data release mechanisms aim to simultaneously minimize information-leakage with respect to sensitive data and distortion with respect to useful data. Dependencies between sensitive and useful data results in a privacy-utility tradeoff that has strong connections to generalized rate-distortion problems. In this work, we study how the optimal privacy-utility tradeoff region is affected by constraints on the data that is directly available as input to the release mechanism. In particular, we consider the availability of only sensitive data, only useful data, and both (full data). We show that a general hierarchy holds: the tradeoff region given only the sensitive data is no larger than the region given only the useful data, which in turn is clearly no larger than the region given both sensitive and useful data. In addition, we determine conditions under which the tradeoff region given only the useful data coincides with that given full data. These are based on the common information between the sensitive and useful data. We establish these results for general families of privacy and utility measures that satisfy certain natural properties required of any reasonable measure of privacy or utility. We also uncover a new, subtler aspect of the data processing inequality for general non-symmetric privacy measures and discuss its operational relevance and implications.  Finally, we derive exact closed-analytic-form expressions for the privacy-utility tradeoffs for symmetrically dependent sensitive and useful data under mutual information and Hamming distortion as the respective privacy and utility measures.
\end{abstract} 


\input{intro}

\input{formulation}

\input{convexity}

\input{privacy}

\input{results}
\input{example}
\input{conclusion}

\bibliographystyle{IEEEtran}
\bibliography{ref}

\appendices

\input{metricProofs}

\input{orderProof}
\input{OPoptProof}
\input{NecessaryCondProof}

\input{lemmaproofs}
\input{xyavailable}

\input{yavailable}

\input{xavailable}

\end{document}

%% file: intro.tex
\section{Introduction}

The objective of privacy-preserving data release is to provide useful
data with minimal distortion while simultaneously minimizing the
sensitive data revealed.  Dependencies between the sensitive and
useful data results in a privacy-utility tradeoff that has strong
connections to generalized rate-distortion
problems~\cite{RebolloFD-TKDE10-tCloseToPRAMviaIT}.  In this work, we
study how the optimal privacy-utility tradeoff region,
for general privacy and distortion measures,
is affected by constraints on the data that is directly available as
input to the release mechanism.  Such constraints are potentially
motivated by applications where either the sensitive or useful data is
not directly observable.  For example, the useful data may be a
latent property that must be inferred from only the sensitive data.
Alternatively, the constraints may be used to capture the limitations
of a particular approach, such as {\em output-perturbation} data
release mechanisms that take only the useful data as input, while
ignoring the remaining sensitive data.

The general challenge of privacy-preserving data release has been the
aim of a broad and varied field of study.  Basic attempts to anonymize
data have led to widely publicized leaks of sensitive information,
such as~\cite{Sweeney-CMU2000-SimplyIdentify,
  NarayananS-SP08-DeAnonNetlix}.  These have subsequently motivated a
wide variety of statistical formulations and techniques for preserving
privacy, such as $k$-anonymity~\cite{Sweeney02-kAnon},
$L$-diversity~\cite{MachanavajjhalaKGV-TKDD07-Ldiversity},
$t$-closeness~\cite{LiLV-07-tCloseness}, and differential
privacy~\cite{DworkMNS-TOC06-DiffPriv}.  
Our work concerns a non-asymptotic, information-theoretic treatment of
this problem, such as in~\cite{RebolloFD-TKDE10-tCloseToPRAMviaIT,
CalmonF-Allerton12}, where the sensitive data and useful data are
modeled as random variables $X$ and $Y$, respectively, and mechanism
design is the problem of constructing channels that obtain the optimal
privacy-utility tradeoffs.
While we consider a non-asymptotic, single-letter problem formulation,
there are also related asymptotic coding problems that additionally
consider communication efficiency in a rate-distortion-privacy
tradeoff, as studied in~\cite{Yamamoto-IT83-RateDistortionSecrecy,
SankarRP-TIFS13-UtilityPrivacy}.

This work makes three main contributions. First, we establish a
fundamental hierarchy of data-release mechanisms in terms of their
privacy-utility tradeoff regions. In particular, we prove that the
tradeoff region given only sensitive data is contained within the
tradeoff region given only useful data. These results are established
for general families of privacy and utility measures that satisfy
certain natural properties required of any reasonable measure of
privacy or utility. Second, we uncover a new, subtler aspect of the
data processing inequality for general non-symmetric privacy measures,
which we term as the linkage inequality, and discuss its operational
relevance and implications. In particular, we show that certain
well-known privacy measures such as maximal information and
differential privacy are not guaranteed to satisfy the linkage
inequality. Third, we derive exact closed-analytic-form expressions
for the privacy-utility tradeoffs for symmetrically dependent
sensitive and useful data under mutual information and Hamming
distortion as the respective privacy and utility measures, for all
three data-release mechanisms that we analyze in this work.

The rest of this paper is organized as follows.
In Sec.~\ref{sec:formulation}, we generalize the framework
of~\cite{RebolloFD-TKDE10-tCloseToPRAMviaIT, CalmonF-Allerton12} to
address arbitrary data observation constraints and general measures for
privacy and utility.  These generalizations allow us to consider
scenarios where the sensitive and useful data are partially
unavailable and/or observed through a noisy channel.  The connections
of this framework to other privacy-utility and generalized
rate-distortion problems encountered in the literature, when
specialized to specific data observation constraints and privacy and
utility measures, are discussed in Sec.~\ref{sec:convexity}.  We also
note that the tradeoff optimization problem with arbitrary observation
constraints is convex if the particular privacy and utility measures
have convexity properties.

In Sec.~\ref{sec:privacy}, we discuss several privacy measures,
including maximal leakage~\cite{Issa16-MaxLeakage} and differential
privacy~\cite{DworkMNS-TOC06-DiffPriv}.  We also examine several basic
properties of these privacy measures and their operational relevance.
A general privacy leakage measure, denoted by $J(X;Z)$, is a functional
of the joint distribution of the sensitive data $X$ and data release
$Z$.  For non-symmetric privacy measures (where $J(X;Z)$ does not
necessarily equal $J(Z;X)$), and given $A \to B \to C$ that form a
Markov chain, the inequality $J(A;C) \leq J(A;B)$ is distinct from
$J(A;C) \leq J(B;C)$.  The first inequality
is equivalent to the well-known post-processing inequality that is
considered an axiomatic requirement of any reasonable privacy
measure~\cite{KiferLin-JPC12-Axioms}.  The second inequality could be
interpreted as bounding privacy leakage for some secondary sensitive
data $A$ when a release mechanism that produces $C$ offers a privacy
leakage guarantee for the primary sensitive data $B$.  Interestingly,
this second inequality does not hold for some privacy measures, such as
differential privacy, and is necessary to show some of our tradeoff
results in Sec.~\ref{sec:hierarchy}.

In Sec.~\ref{sec:hierarchy}, we compare the optimal privacy-utility
tradeoffs under three scenarios, where only the
sensitive data, only the useful data, or both (full data) are
available.  We show that a general hierarchy holds, that is, the
tradeoff region given only the sensitive data is no larger than the
region given only the useful data, which in turn is clearly no larger
than the region given both sensitive and useful data.  We also show
that if the common information and mutual information between the
sensitive and useful data are equal\footnote{This statement applies
for both the Wyner~\cite{Wyner-75-CommonInfo} and
G\'acs-K\"orner~\cite{GacsKorn-73-CommonInfo} notions of common
information.}, then the tradeoff region given only the useful data
coincides with that given full data, indicating when output
perturbation is optimal despite unavailability of the sensitive data.
Conversely, when the common information and mutual information are not
equal, there exist distortion measures where the tradeoff regions are
not the same, indicating that output perturbation can be strictly
suboptimal compared to the full data scenario.
In Sec.~\ref{sec:example}, we present an example with analytically derived
optimal privacy-utility tradeoffs illustrating the hierarchy established
by the results in Sec.~\ref{sec:hierarchy}.

%% file: formulation.tex
\section{Privacy-Utility Tradeoff Problem}
\label{sec:formulation}

\begin{figure}[t]
   \centering
   \includegraphics[width=0.48\textwidth]{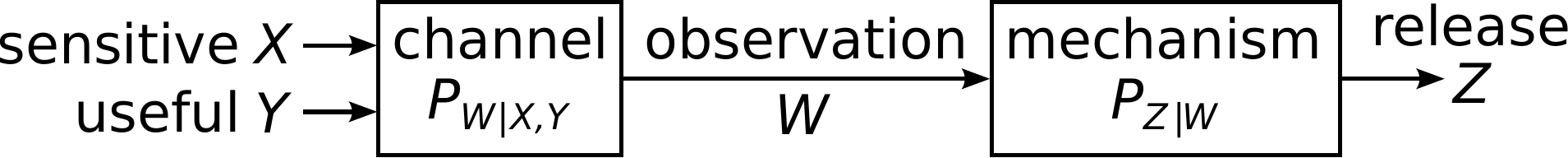}
   \caption{The observation $W$ of the sensitive data $X$ and useful data $Y$ is input to the data release mechanism which produces the released data $Z$.}
   \label{fig:problemDiagram}
 \end{figure}

Let $X$, $Y$, and $W$ be discrete random variables (RVs) distributed
on finite alphabets $\mathcal{X}$, $\mathcal{Y}$ and $\mathcal{W}$,
respectively. Let $X$ denote the sensitive information that the user
wishes to conceal, $Y$ the useful information that the user is willing
to reveal, and $W$ the directly observable data, which may represent a
noisy observation of $X$ and/or $Y$.  The target application
imposes the specific {\em data model} $P_{X,Y}$
and {\em observation constraints} $P_{W|X,Y}$ so that $(X, Y, W) \sim
P_{X,Y} P_{W|X,Y}$.  The {\em data release mechanism} takes $W$ as input
and (randomly) generates output $Z$ in a given finite alphabet
$\mathcal{Z}$ dictated by the target application (perhaps implicitly
via the distortion measure).  Note that $(X,Y) \to W \to Z$ form a 
Markov chain and the mechanism can be specified by the
conditional distribution $P_{Z|W}$. A diagram of the overall system is shown in
Figure~\ref{fig:problemDiagram}.

The mechanism should be designed such that $Z$ provides
application-specific utility through the information it reveals about
$Y$ while protecting privacy by limiting the information it reveals
about $X$.

{\bf Privacy:} The privacy of the mechanism-output $Z$ is
inversely quantified by a general privacy-leakage measure $J(X;Z)$,
which is a functional\footnote{Formally, the privacy measure
notation should be $J(P_{X,Z})$, but for convenience we adopt
$J(X;Z)$, an abuse of notation similar to the use of $I(X;Y)$ for
mutual information.} that assigns values in $[0, \infty)$ to
joint distributions of $X$ and $Z$.  Thus, the aim of privacy is to
minimize $J(X;Z)$, which ideally becomes perfect when $J(X;Z) = 0$.
The privacy-leakage measure need not be symmetric, i.e.,
$J(X;Z)$ need not equal $J(Z;X)$.
Examples of privacy measures include symmetric ones like {\em mutual
information}, where $J(X;Z) = I(X;Z)$, which captures an
average-case information leakage, and asymmetric ones like {\em maximal
information leakage}, where $J(X;Z) = \max_{z \in \mathcal{Z}} H(X)
- H(X | Z = z)$ \cite{CalmonF-Allerton12}. In Sec.~\ref{sec:privacy}
we will discuss three other privacy measures: information privacy,
differential privacy, and Sibson mutual information. The first of
these is symmetric, while the other two are not.

{\bf Utility:} The amount of utility that the
mechanism-output $Z$ provides about the useful information represented
by $Y$ is inversely quantified by a general distortion measure
$D(P_{Y,Z})$, which is a functional that assigns values in $[0, \infty)$
to joint distributions of $Y$ and $Z$.
Thus, the aim is to minimize $D(P_{Y,Z})$.
As in the case of privacy, distortion measures need not be symmetric.
The specific distortion measure is dictated by the target
application.  Example distortion measures include: 1) {\em expected
distortion}, where $D(P_{Y,Z}) = E[d(Y,Z)]$ for some distortion
function $d: \mathcal{Y} \times \mathcal{Z} \to [0, \infty)$, and 2) {\em
conditional entropy}, where $D(P_{Y,Z}) = H(Y|Z)$ which corresponds
to the goal of maximizing the mutual information between $Y$ and
$Z$.  Note that probability of error $\Pr(Y \neq Z)$ is an example
within the class of expected distortion measures where $d(y,z)$ is
equal to zero when $y = z$ and equal to one otherwise.

{\bf Privacy-utility tradeoff:} Given a target application
that specifies the data model $P_{X,Y}$, observation model
$P_{W|X,Y}$, and distortion measure $D(P_{Y,Z})$, the goal of the
system designer is to construct mechanisms $P_{Z|W}$ that provide the
desired levels of privacy and utility while achieving the optimal
tradeoff.  We say that a particular privacy-utility pair $(\epsilon,
\delta) \in [0, \infty)^2$ is {\em achievable} if there exists a
mechanism $P_{Z|W}$ with privacy leakage $J(X;Z) \leq \epsilon$ and
distortion $D(P_{Y,Z}) \leq \delta$. The set of all achievable
privacy-utility pairs forms the {\em achievable region} of
privacy-utility tradeoffs.  We are particularly interested the {\em
optimal boundary} of this region, which can be expressed by the
optimization problem
\begin{equation} \label{eqn:general_tradeoff_opt}
\begin{aligned}
\pi(\delta) := & \inf_{P_{Z|W}} J(X;Z) \\
& \ \ \text{s.t.} \ D(P_{Y,Z}) \leq \delta,
\end{aligned}
\end{equation}
which determines the optimal privacy leakage as a function of the
allowable distortion $\delta$.

The distortion constraint, $D(P_{Y,Z}) \leq \delta$, can be
equivalently expressed as a constraint on the conditional distribution
$P_{Z|Y}$ since $P_Y$ is fixed by the data model.  Note that a
mechanism specified by $P_{Z|W}$ determines the corresponding
$P_{Z|Y}$ through the linear relationship\footnote{This and all other statements involving conditional distributions are defined only for symbols in the support of the conditioned random variables.}
\begin{equation} \label{eqn:ZgivenYlinear}
\! P_{Z|Y}(z|y) = \hspace{-12pt} \sum_{w \in \mathcal{W}, x \in \mathcal{X}} \hspace{-12pt}
P_{Z|W}(z|w) P_{W|X,Y}(w|x,y) P_{X|Y}(x|y).
\end{equation}
Similarly, $P_{Z|X}$ is determined by $P_{Z|W}$ through the linear
relationship
\begin{equation} \label{eqn:ZgivenXlinear}
\! P_{Z|X}(z|x) = \hspace{-12pt} \sum_{w \in \mathcal{W}, y \in \mathcal{Y}} \hspace{-12pt}
P_{Z|W}(z|w) P_{W|X,Y}(w|x,y)P_{Y|X}(y|x).
\end{equation}

While general observation models $P_{W|X,Y}$ can be considered within this
framework, particular structures may be of interest for certain
applications.  We highlight and explore the relationship between three
specific cases for $W$, while allowing a general distribution $P_{X,Y}$
between the sensitive and useful data.

1) {\it Full Data}: In this case, $P_{X,Y}$ is general but
$W = (X,Y)$, capturing the situation when the mechanism has direct
access to both the sensitive and useful information.  For this case,
the privacy-utility optimization problem
of~\eqref{eqn:general_tradeoff_opt} reduces to
\begin{equation}\label{eqn:FD_opt_problem}
\begin{aligned}
\pi_\text{FD}(\delta) := & \inf_{P_{Z|X,Y}} J(X;Z) \\
& \ \ \text{s.t.} \ D(P_{Y,Z}) \leq \delta.
\end{aligned}
\end{equation}

2) {\it Output Perturbation}: In this case, $P_{X,Y}$ is
general but $W = Y$, capturing the situation when the mechanism only
has direct access to the useful information.  For this case, the
privacy-utility optimization problem
of~\eqref{eqn:general_tradeoff_opt} reduces to
\begin{equation}\label{eqn:OP_opt_problem}
\begin{aligned}
\pi_\text{OP}(\delta) := & \inf_{P_{Z|Y}} J(X;Z) \\
& \ \ \text{s.t.} \ D(P_{Y,Z}) \leq \delta,
\end{aligned}
\end{equation}
where $P_{Z|X}(z|x) = \sum_{y \in \mathcal{Y}} P_{Z|Y}(z|y)
P_{Y|X}(y|x)$.
Note that this optimization is equivalent to that of~\eqref{eqn:FD_opt_problem},
with the Markov chain $X \to Y \to Z$ imposed as an additional
constraint.

3) {\it Inference}: In this case, $P_{X,Y}$ is general but
$W = X$, capturing the situation when the mechanism only has direct
access to the sensitive information, but the useful information, such
as a hidden state, is not directly observable and needs to be {\it
inferred} indirectly by processing the sensitive information.
For this case, the privacy-utility
optimization problem of~\eqref{eqn:general_tradeoff_opt} reduces to
\begin{equation}\label{eqn:INF_opt_problem}
\begin{aligned}
\pi_\text{INF}(\delta) := & \inf_{P_{Z|X}} J(X;Z) \\
& \ \ \text{s.t.} \ D(P_{Y,Z}) \leq \delta,
\end{aligned}
\end{equation}
where $P_{Z|Y}(z|y) = \sum_{x \in \mathcal{X}} P_{Z|X}(z|x)
P_{X|Y}(x|y)$.
Note that this optimization is equivalent to that of~\eqref{eqn:FD_opt_problem}, with
the Markov chain $Y \to X \to Z$ imposed as an additional constraint.

%% file: convexity.tex
\section{Convexity and Rate-Distortion Connections}
\label{sec:convexity}

Here we discuss how under certain combinations of data constraints and privacy and 
utility measures, the tradeoff optimization of~\eqref{eqn:general_tradeoff_opt}
specializes to various rate-distortion and privacy-utility
problems encountered in the literature.  We also indicate when the
general tradeoff optimization of~\eqref{eqn:general_tradeoff_opt}
becomes convex for particular privacy and utility measures.

Recall that the distributions $P_{Z|X}$ and $P_{Z|Y}$ are {\it linear} functions of the optimization variable
$P_{Z|W}$ as shown by~\eqref{eqn:ZgivenYlinear}
and~\eqref{eqn:ZgivenXlinear}, while $P_{X,Y,W}$ and its marginals are fixed.
Thus, the convexity properties of the general problem (and in the three scenarios given
by~\eqref{eqn:FD_opt_problem},~\eqref{eqn:OP_opt_problem},
and~\eqref{eqn:INF_opt_problem}) will follow from
the convexity properties of the privacy and distortion measures as
functions of $P_{Z|X}$ and $P_{Z|Y}$, respectively.
For example, with mutual information as the privacy measure $I(X;Z)$, the objective of the tradeoff optimization problem is a convex functional of $P_{Z|X}$.
Any distortion measure that is a convex functional of $P_{Z|Y}$ results in a convex constraint.
For example, any expected distortion utility measure
$D(P_{Y,Z}) = E[d(Y,Z)]$ is a linear (and hence convex)
functional of $P_{Z|Y}$.

The privacy-utility tradeoff problem as considered
by~\cite{RebolloFD-TKDE10-tCloseToPRAMviaIT, CalmonF-Allerton12}
assumes the output perturbation constraint
(see~\eqref{eqn:OP_opt_problem}), while using expected distortion
$D(P_{Y,Z}) = E[d(Y,Z)]$ as the utility measure, and mutual information
$I(X;Z)$ as the privacy measure.
Additionally,~\cite{CalmonF-Allerton12} also considers maximum
information leakage, $\max_{z \in \mathcal{Z}} \left[ H(X) - H(X | Z =
z) \right]$, as an alternative privacy measure.  As noted
by~\cite{CalmonF-Allerton12}, the optimization problem for the full
data scenario (see~\eqref{eqn:FD_opt_problem}) can be recast as an
optimization with the output perturbation constraint, by redefining
the useful data as $Y' := (X,Y)$ and the distortion function as
$d'(Y',Z) := d(Y,Z)$.  This approach allows one to solve the
optimization problem for the full data scenario using an equivalent
optimization problem appearing in the output perturbation scenario.
However, the distinction between these two scenarios should not be
overlooked, as the output perturbation scenario represents a
fundamentally different problem where the sensitive data is not
available, which in general results in a strictly smaller
privacy-utility tradeoff region (see Theorem~\ref{theorem:OPConverse}).

The inference scenario given by~\eqref{eqn:INF_opt_problem} with mutual information
as the privacy measure and expected distortion $D(P_{Y,Z}) = E[d(Y,Z)]$ as the utility measure is
equivalent to an indirect rate-distortion problem~\cite{witsenhausen1980indirect}.
As shown by Witsenhausen in~\cite{witsenhausen1980indirect}, indirect
rate-distortion problems can be converted to direct ones with the
modified distortion measure $d'(x,z) := E[d(Y,Z)|X = x, Z = z] =
\sum_{y\in \mathcal{Y}} d(y,z) P_{Y|X}(y|x)$ since $Y \to X \to Z$
forms a Markov chain.

When the utility measure is conditional entropy, i.e., $D(P_{Y,Z}) =
H(Y|Z)$, the distortion constraint can be equivalently
written as $I(Y;Z) \geq \delta'$, where $\delta' := H(Y) - \delta$,
thus the utility objective is to maximize the mutual information $I(Y;Z)$.
Combining this with mutual information as the privacy measure
results in the optimization problem of choosing
$Z$ to minimize $I(X;Z)$ subject to a lower bound on $I(Y;Z)$.  This
problem in the inference scenario, where the additional Markov chain
constraint $Y \to X \to Z$ is imposed, is equivalent to the
Information Bottleneck problem considered
in~\cite{TishbyPB-Allerton00-InfoBottleneck}, which also provides a
generalization of the Blahut-Arimoto algorithm~\cite{cover2012elements}
to perform this optimization.  For the output
perturbation scenario, where instead the Markov chain constraint $X
\to Y \to Z$ is imposed, this problem is called the Privacy Funnel and
was proposed by~\cite{MakhdoumiSFW-ITW14-PrivacyFunnel}.  In all three
scenarios, the optimization problems are non-convex as the feasible
regions are non-convex, and specifically are complements of convex
regions.

%% file: privacy.tex
\section{Privacy Measures and Properties}
\label{sec:privacy}

We allow general statistical measures of privacy-leakage
that can be arbitrary functionals of the joint
distribution between the sensitive data $X$ and the
release $Z$. However, in order for some of our later
results in Section~\ref{sec:hierarchy} to hold, the privacy
measure must posses certain natural, desirable properties
described in this section. In particular, generalized
analogies of the data processing inequality are
important. We will also discuss several privacy measures
encountered in the literature and whether they satisfy these
properties.

We will generally assume the following two properties,
which hold for all of the specific privacy measures
discussed in this paper.

In this section, we focus on privacy measures in more
detail and generality. We discuss certain key desirable properties
that any measure of privacy should satisfy within the context of
privacy-preserving data release. In particular, generalized
analogies of the data processing inequality are
important. Specifically, we uncover and highlight a new, subtler
aspect of the data processing inequality for general non-symmetric
privacy measures, which we term as the linkage inequality, and
discuss its operational relevance and implications. We show that
certain well-known privacy measures such as maximal information and
differential privacy are not guaranteed to satisfy the linkage
inequality. Our results pertaining to the fundamental hierarchy of
privacy-utility tradeoffs in Sec.~\ref{sec:hierarchy} hold for
general privacy measures that satisfy the properties described in
this section.

We allow general statistical measures of privacy-leakage
that can be arbitrary functionals of the joint distribution between
the sensitive data $X$ and the release $Z$. However, we require that
the privacy measure satisfy the following two basic properties which
hold for all of the specific privacy measures discussed in this
paper.
\begin{itemize}
\item \textbf{Perfect privacy is independence:} $J(X;Z) \geq 0$ with
  equality if and only if $X$ and $Z$ are independent.
\item \textbf{Privacy invariance:} $J(X_1;Z_1) = J(X_2;Z_2)$ if
  $P_{X_1, Z_1}$ and $P_{X_2, Z_2}$ are isomorphically equivalent
  distributions.
\end{itemize}

The following property establishes that a privacy measure captures the notion that
privacy cannot be worsened, i.e., privacy-leakage
cannot be increased, by independent post-processing of the released
data.  This well-known concept is considered a fundamental, axiomatic
requirement for any reasonable privacy measure~\cite{KiferLin-JPC12-Axioms}.

\begin{definition} \label{def:PostProcIneq}
\textbf{(Post-processing inequality)} A privacy measure $J$ satisfies
the post-processing inequality if and only if for any $A \to B \to C$
that form a Markov chain, we have that $J(A;B) \geq J(A;C)$.
\end{definition}

For symmetric privacy measures where $J(X;Z) = J(Z;X)$ (i.e.,
privacy-leakage remains unchanged when swapping the roles of the
release and sensitive data), the next property is equivalent to the
post-processing inequality.  However, for asymmetric privacy measures,
this property is a distinct concept.

\begin{definition} \label{def:LinkageIneq}
\textbf{(Linkage inequality)} A privacy measure $J$ satisfies the
linkage inequality if and only if for any $A \to B \to C$ that form a
Markov chain, we have that $J(B;C) \geq J(A;C)$.
\end{definition}

The linkage inequality captures the notion that if there were primary
and secondary sensitive data and the release was independently
generated from only the primary sensitive data, then the
privacy-leakage for the secondary sensitive data is bounded by the
privacy-leakage for the primary sensitive data.  Intuitively, this
concept corresponds to the privacy-leakage of the secondary sensitive
data occurring via and being limited by the privacy-leakage of the
primary sensitive data.  Pragmatically, this property allows for
convenient bounds when making privacy guarantees, especially when
there may be unforeseen secondary sensitive data correlated to the
primary sensitive data considered.

Note that satisfying both inequalities of
Definitions~\ref{def:PostProcIneq} and~\ref{def:LinkageIneq} would
imply the property of privacy invariance assumed earlier, but the
reverse is not necessarily true.  Of course, when mutual information
is the privacy measure, both of these inequalities are immediate as
they are equivalent to the data processing inequality.

In the rest of this section, we discuss the post-processing
and linkage inequalities in the context of a number of commonly
encountered privacy measures.

\subsection{Maximal Information Leakage}

The {\it maximal information leakage} measure, introduced
in~\cite{CalmonF-Allerton12}, is defined as follows
\begin{equation} \label{eqn:maxInfoLeakage}
I^*(X; Z) := H(X) - \min_{z \in \mathcal{Z}} H(X | Z = z),
\end{equation}
This is an example of an {\it asymmetric} privacy measure that aims to
capture the worst-case information leakage over the possible releases.
Interestingly, while the post-processing inequality holds for this
measure, the linkage inequality does not.
The proof of this proposition is given in Appendix~\ref{app:MaxIPineq}.

\begin{proposition} \label{prop:MaxIPineq}
The maximal information leakage measure $I^*(X; Z)$ satisfies the
post-processing inequality, but does not satisfy the linkage
inequality.
\end{proposition}

Note that swapping the roles of $X$ and $Z$ to define $J(X;Z) = I^*(Z;X)$
would yield a measure that satisfies the linkage inequality, but not the
post-processing inequality.

\subsection{Maximal Leakage via Sibson Mutual Information}

Another privacy measure similarly called {\em maximal leakage} is equivalent to Sibson mutual information of order infinity~\cite{sibson1969information}, which is given by
\[
I_\infty(X;Z) := \log \sum_{z \in \mathcal{Z}} \max_{x: P_X(x) > 0} P_{Z|X}(z|x).
\]
Demonstrating its operational significance as a privacy measure,~\cite{Issa16-MaxLeakage} showed that
\[
I_\infty(X;Z) = \sup_{U \to X \to Z \to \hat{U}} \log \frac{\Pr(\hat{U} = U)}{\max_u P_U(u)},
\]
which implies that $\exp(I_\infty(X;Z))$ bounds the multiplicative advantage gained from observing $Z$ for guessing any (potentially random) function of $X$.
This operational bound holds even for generalizations allowing multiple or approximate guesses (see details in~\cite{Issa16-MaxLeakage}).
Maximal leakage is asymmetric and satisfies the post-processing and linkage inequalities~\cite{Issa16-MaxLeakage}.

\subsection{Information Privacy}

The {\it information privacy} (IP) measure was introduced
in~\cite{CalmonF-Allerton12}. The following definition differs from
the one given in \cite{CalmonF-Allerton12}, but is equivalent to
it (see Corrolary~\ref{cor:InfoPrivEquiv}),
\begin{equation} \label{eqn:epsInfoPriv}
IP(X;Z) := \max_{x,z : P_X(x), P_Z(z) > 0} \left | \ln
\frac{P_{X,Z}(x,z)}{P_X(x) P_Z(z)} \right |,
\end{equation}
where we adopt the convention that $|\ln 0| = \infty$, denoting that
IP leakage is unbounded when there exist $x$ and $z$ such that
$P_X(x), P_Z(z) > 0$ and $P_{X,Z}(x,z) = 0$.  This quantity can be
equivalently viewed as a bound on the absolute log-ratio of the
sensitive data prior distribution and the posterior
distribution given the release, since
\[
\frac{P_{X,Z}(x,z)}{P_X(x) P_Z(z)} = \frac{P_{X|Z}(x|z)}{P_X(x)}.
\]
With respect to the definition of information privacy in
\cite{CalmonF-Allerton12}, a data release mechanism $P_{Z|X}$
provides $\epsilon$-information privacy if $IP(X;Z) = \epsilon$.

\begin{lemma} \label{lem:EpsInfoPrivIneq}
The information privacy measure $IP(X;Z)$ satisfies both the post-processing and
linkage inequalities.
\end{lemma}

Lemma~\ref{lem:EpsInfoPrivIneq} leads to the following
corollary which implies that expanding the domain of maximization
in~\eqref{eqn:epsInfoPriv}, from singleton events $\{x\}$ and
$\{z\}$ to events $A \subset \mathcal{X}$ and $B \subset
\mathcal{Z}$, does not increase the maximum value.

\begin{corollary} \label{cor:InfoPrivEquiv}
The information privacy measure is equivalently given by
\begin{equation*}
IP(X;Z) = \max_{\substack{\mathcal{A}\subseteq \mathcal{X},
    \mathcal{B} \subseteq \mathcal{Z}:\\ \Pr(X \in \mathcal{A}), \Pr(Z \in
    \mathcal{B}) > 0}} \left | \ln \frac{\Pr(X\in \mathcal{A}, Z\in
  \mathcal{B})}{\Pr(X\in \mathcal{A}) \Pr(Z \in \mathcal{B})} \right |.
\end{equation*}
\end{corollary}

\noindent The proofs of Lemma~\ref{lem:EpsInfoPrivIneq} and
Corollary~\ref{cor:InfoPrivEquiv} are presented in
Appendices~\ref{app:EpsInfoPrivIneq} and \ref{app:InfoPrivEquiv}
respectively.

\subsection{Differential Privacy}

The {\em differential privacy} (DP) measure was introduced
by~\cite{DworkMNS-TOC06-DiffPriv} and has been extensively studied
in the context of privacy-preserving querying of databases.
For ease of exposition, within this subsection we will model a database as
a length-$n$ binary sequence, i.e., in the domain $\mathcal{X} =
\{0,1\}^n$, and assume a discrete release alphabet $\mathcal{Z}$.
However, the concepts and discussion readily generalize.

\begin{definition}
A release mechanism $P_{Z|X}$ with domain $\mathcal{X} = \{0,1\}^n$
and range $\mathcal{Z}$ is $\epsilon$-differentially
private if for all $\mathcal{B} \subseteq \mathcal{Z}$ and $x_1,
x_2 \in \mathcal{X}$ such that $d_H(x_1,x_2) \leq 1$, where $d_H$
denotes Hamming distance, we have
\[
\Pr(Z \in \mathcal{B} | X = x_1) \leq e^{\epsilon} \cdot \Pr(Z\in \mathcal{B} |
X = x_2).
\]
\end{definition}

Implicitly, if there exist $x_1,x_2 \in \mathcal{X}$ with
$d_H(x_1,x_2) = 1$ and $z \in \mathcal{Z}$ such that $P_{Z|X}(z|x_1) >
0$, but $P_{Z|X}(z|x_2) = 0$, then the release mechanism $P_{Z|X}$ is
not differentially private for any $\epsilon$.  The differential
privacy measure $DP(X;Z)$ is defined as the smallest value of
$\epsilon$ for which $P_{Z|X}$ is $\epsilon$-differentially private,
which is expressed in the following lemma whose proof
is presented in Appendix~\ref{app:DPmetric}.

\begin{lemma} \label{lem:DPmetric}
The differential privacy measure is given by
\[
DP(X;Z) = \max_{\substack{x_1,x_2 \in \mathcal{X}, z \in \mathcal{Z}:\\
  d_H(x_1,x_2) = 1}}
\left | \ln \frac{P_{Z|X}(z|x_1)}{P_{Z|X}(z|x_2)} \right |,
\]
where we adopt the conventions that $|\ln (c/0)| = |\ln 0|= \infty$ and $|\ln (0/0)| = 0$.
\end{lemma}

It is well-known that $DP(X;Z)$ satisfies the post-processing
inequality~\cite{KiferLin-JPC12-Axioms}. However, we demonstrate via an
example that $DP(X;Z)$ does not satisfy the linkage inequality. This
has important philosophical implications on the use of differential
privacy which we then discuss.

\begin{proposition} \label{prop:DPexample}
The differential privacy measure $DP(X;Z)$ does not satisfy the linkage inequality.
\end{proposition}

The proof of Proposition~\ref{prop:DPexample} (see Appendix~\ref{app:DPexample}) constructs a simple example with 
databases $A, B \in \{0,1\}^2$, where $B := (B_1,B_2)$ is a deterministic function
of the database $A := (A_1,A_2)$, given by $B_1 = B_2 = A_1 \vee A_2$.
This example could be interpreted as a toy model for the spread of a contagious disease between two close relatives, where $A$ denotes the infection status of each person at an earlier time and $B$ at a later time, while simply depicting inevitable disease transmission.
The proof then constructs an example mechanism $P_{C|B}$ that when applied to $B$ (such that $A \to B \to C$ forms a Markov chain), we have $DP(A;C) > DP(B;C)$ showing violation of the linkage inequality.

More generally, the consequences of not satisfying the linkage inequality can impact situations where a dataset has been vertically partitioned over two tables $A$ and $B$ (each containing different attributes of the same population), or when a table $A$ is preprocessed to produce table $B$.
A differentially private release mechanism $P_{C|B}$ applied to the table $B$ may not guarantee the same level of privacy with respect to the potentially sensitive data in table $A$.
Since the effective release mechanism (overall channel) from $A$ to $C$ is given by $P_{C|A}(c|a) = \sum_{b \in \mathcal{B}} P_{C|B}(c|b) P_{B|A}(b|a)$, correlation across data tuples (as introduced by $P_{B|A}$) may cause $P_{C|A}$ to be less differentially private than $P_{C|B}$.
This realization is related to broader observations on the impact of data correlation on differential privacy guarantees and susceptibility to inference attacks (see~\cite{KiferM-2011-NoFreeLunch, LiuCM-2016-DDP} and references therein).

%% file: results.tex
\section{Hierarchy of Privacy-Utility Tradeoffs under Data Constraints}
\label{sec:hierarchy}

In this section we establish a fundamental hierarchy for
data-release mechanisms in terms of their privacy-utility tradeoff
regions. In particular, we prove that the tradeoff region given only
sensitive data is contained within the tradeoff region given only
useful data.

For a given (fixed) distribution $P_{X,Y}$ between the sensitive and
private data, we can study how the optimal privacy-utility tradeoff
changes across the aforementioned three different cases of $W$.  This
is of practical interest, since the restrictions on $W$ in the
inference and output perturbation scenarios might be considered not
just when these situations inherently arise in the
given application, but also for simplifying mechanism design and
optimization.

Since the optimization problems of~\eqref{eqn:OP_opt_problem}
and~\eqref{eqn:INF_opt_problem} are equivalent to
\eqref{eqn:FD_opt_problem} with an additional Markov chain constraint,
we immediately have that $\pi_\text{FD}(\delta) \leq
\pi_\text{OP}(\delta)$ and $\pi_\text{FD}(\delta) \leq
\pi_\text{INF}(\delta)$ for any $\delta$.  This implies that the
achievable privacy-utility regions of both the inference scenario and
output perturbation scenario are contained within the achievable
privacy-utility region of the full data scenario, which intuitively
follows since in the full data scenario only more input data
available. The next theorem establishes the general relationship
between the inference and output perturbation tradeoff regions.

\begin{theorem} \label{theorem:order}
\textbf{(Output Perturbation better than Inference)} For any data
model $P_{X,Y}$, distortion measure $D(P_{Y,Z})$, and privacy measure
$J(X;Z)$ that satisfies the linkage inequality,
the achievable privacy-utility region for the output perturbation scenario (when $W
= Y$) contains the achievable privacy-utility region for the inference
scenario (when $W = X$), that is, $\pi_\text{OP}(\delta) \leq
\pi_\text{INF}(\delta)$ for any $\delta$.
\end{theorem}

\noindent The proof of Theorem~\ref{theorem:order} is presented in Appendix~\ref{app:OrderProof}.

Combining the preceding theorem with the earlier observations, we have
that $\pi_\text{FD}(\delta) \leq \pi_\text{OP}(\delta) \leq
\pi_\text{INF}(\delta)$ for any $\delta$. Thus, in general, full data
offers a better privacy-utility tradeoff than output perturbation,
which in turn offers a better privacy-utility tradeoff than inference.

The next theorem establishes that for a certain class of joint
distributions $P_{X,Y}$, the full data and output perturbation
scenarios have the same optimal privacy-utility tradeoff.  Thus,
for this class of $P_{X,Y}$, the full data mechanism design can be
simplified to the design of an output perturbation mechanism, which can
ignore the sensitive data $X$ without degrading the privacy-utility
performance.  Specifically, this class is characterized by those joint
distributions $P_{X,Y}$ for which common information $C(X;Y) = I(X;Y)$.
Some of the key properties of common information that are needed for proving
Theorems~\ref{theorem:OPoptimality} and~\ref{theorem:OPConverse}
are summarized in Appendix~\ref{app:CommonInfoProps}.

\begin{theorem} \label{theorem:OPoptimality}
\textbf{(Sufficient Conditions for the Optimality of Output Perturbation)}
For any distortion measure $D(P_{Y,Z})$, any privacy measure $J(X;Z)$ that satisfies the linkage inequality, and
any data model $P_{X,Y}$ where $C(X;Y) = I(X;Y)$, the achievable
privacy-utility region for the output perturbation scenario (when $W
= Y$) is the same as the achievable privacy-utility region for the full
data scenario (when $W = (X,Y)$), that is, $\pi_\text{OP}(\delta) =
\pi_\text{FD}(\delta)$ for any distortion measure and any $\delta$.
\end{theorem}

\noindent The proof of Theorem~\ref{theorem:OPoptimality} is presented in Appendix~\ref{app:OPoptimalityProof}.

Theorem~\ref{theorem:OPoptimality} establishes that $C(X;Y) = I(X;Y)$ is a sufficient condition on $P_{X,Y}$ such that, for any general distortion measure, full data mechanisms cannot provide better privacy-utility tradeoffs than the output perturbation mechanisms.
Our next theorem gives the converse result, establishing that for data models where $C(X;Y) \neq I(X;Y)$, output perturbation mechanisms are generally suboptimal, that is, there exists a distortion measure such that the full data mechanisms provide a strictly better privacy-utility tradeoff.

\begin{theorem} \label{theorem:OPConverse}
\textbf{(Necessary Conditions for the Optimality of Output Perturbation)}
For any data model $P_{X,Y}$ where $C(X;Y) \neq I(X;Y)$,
there exists a distortion measure $D(P_{Y,Z})$ such that the achievable
privacy-utility region for the output perturbation scenario (when $W
= Y$) is strictly smaller than the achievable privacy-utility region
for the full data scenario (when $W = (X,Y)$), that is, there exists
$\delta \geq 0$ such that $\pi_\text{OP}(\delta) >
\pi_\text{FD}(\delta)$.
\end{theorem}

\noindent The proof of Theorem~\ref{theorem:OPConverse} is presented in Appendix~\ref{app:OPConverseProof}.

%% file: example.tex
\section{Analytical Privacy-Utility Tradeoff Examples}
\label{sec:example}

In this section, we consider an example data model $P_{X,Y}$ and analytically derive the optimal privacy-utility tradeoffs under the full data, output perturbation, and inference scenarios.
For this example, we use mutual information as the privacy measure and probability of error as the distortion measure, i.e., $J(X;Z) = I(X;Z)$ and $D(P_{Y,Z}) = \Pr(Y \neq Z)$, where $Z$ is the released data.
Our particular toy data model assumes that the sensitive data $X$ and useful data $Y$ are discrete random variables on the same finite set $\mathcal X = \mathcal Y = \{0,\ldots,m-1\}$, with the joint distribution
\begin{align}\label{symmetric_pmf1} 
P_{X,Y}(x,y) = 
\begin{cases}
\frac{1-p}{m}, & \text{if } x=y, \\
\frac{p}{m (m-1)}, & \text{otherwise},
\end{cases}
\end{align}
where the distribution parameters $p \in [0,1]$ and $m \in \mathbb{Z}$ with $m \geq 2$.
We will call the joint distribution in~\eqref{symmetric_pmf1} the \emph{symmetric pair} and use the notation $(X,Y) \sim SP(m, p)$.

The symmetric pair distribution can be viewed as a generalization of the binary symmetric source to an $m$-ary alphabet.
The parameter $p$ is analogous to the cross-over probability and equal to $\Pr(X \neq Y)$.
Note that both $X$ and $Y$ are marginally uniform and that the joint distribution could be equivalently defined via the channel
\begin{align}\nonumber
Y = X + N \mod m,
\end{align}
where $N \in \{0, \ldots, m-1\}$ is independent additive noise with the distribution
\begin{align}
\label{eqn:SP_noise}
P_{N}(n) = 
\begin{cases}
1-p, & \text{if } n=0 \\
\frac{p}{m-1}, & \text{otherwise.}
\end{cases}
\end{align}

The mutual information of the symmetric pair distribution, which we denote as a function $r_m(p)$ of the distribution parameters $m$ and $p$, is given by the next lemma and used extensively in the tradeoff results and proofs.

\begin{lemma}
\label{lemma:sp_mut_computation}
\textbf{(Mutual Information of Symmetric Pair)}
If $(X,Y) \sim SP(m,p)$, then
\begin{equation*}
I(X;Y) = \log m - p \log(m-1) - h_2(p) =: r_m(p),
\end{equation*}
where $h_2(p) := -p \log p - (1-p) \log (1-p)$ is the binary entropy function.
\end{lemma}

\begin{proof} We have that
\begin{align*}
I(X;Y) &= H(Y) - H(Y|X) \\
&= H(Y) - H(N) \\
&= \log m + (1-p) \log (1-p) + p \log \frac{p}{m-1} \\
&= \log m - p \log(m-1) - h_2(p),
\end{align*}
where $N$ is independent noise given by~\eqref{eqn:SP_noise}.
\end{proof}

For our example data model, the next three theorems provide the analytically derived optimal privacy-utility tradeoffs under the full data, output perturbation, and inference scenarios.
Note that for any distortion constraint $\delta \geq 1-\frac{1}{m}$, we can immediately achieve perfect privacy, i.e., $\pi_\text{FD}(\delta) = \pi_\text{INF}(\delta) = \pi_\text{OP}(\delta) = 0$, via the mechanism that trivially releases $Z$ that is independent of $(X,Y)$ and uniform over $\mathcal{Y}$, which obtains distortion $\Pr(Y \neq Z) = 1-\frac{1}{m} \leq \delta$ and perfect privacy $I(X;Z) = 0$.

\begin{theorem} \label{thm:xy_available}
\textbf{(Full Data Privacy-Utility Tradeoff for the Symmetric Pair Distribution)}
With mutual information as the privacy measure, $J(X;Z) = I(X;Z)$, and probability of error as the distortion measure, $D(P_{Y,Z}) = \Pr(Y \neq Z)$, if the data model is $(X,Y) \sim SP(m,p)$, then
the optimal privacy-utility tradeoff for the full data scenario in~\eqref{eqn:FD_opt_problem} is given by
\begin{align} \label{eqn:SP_FD_region}
\pi_\text{FD}(\delta) =
\begin{cases}
r_m(p + \delta), & \text{if } \delta \leq 1-\frac{1}{m} - p, \\
r_m(p - \delta), & \text{if } \delta \leq p - (1 - \frac{1}{m}), \\
0, & \text{otherwise}.
\end{cases}
\end{align}
For $p \leq 1-\frac{1}{m}$, the optimal mechanism $P_{Z|X,Y}$ is defined by
\begin{align}\label{remark:general_case_solving_rv1}
Z :=
\begin{cases}
Y + N \mod m, & \text{if } X = Y,\\
Y, & \text{otherwise},
\end{cases}
\end{align}
where $N \in \{0, \ldots, m-1\}$ is independent of $(X,Y)$ with the distribution
\begin{align*}
P_{N}(n) = 
\begin{cases}
1-\frac{t}{1-p}, & \text{if } n = 0 \\
\frac{t}{(1-p)(m-1)}, & \text{otherwise},
\end{cases}
\end{align*}
where $t:= \min(1-\frac{1}{m}-p,\delta)$.
\end{theorem}

\noindent The proof of Theorem~\ref{thm:xy_available}
is presented in Appendix~\ref{app:proof_xy_available}.

Observe that in the case of $p \leq 1-\frac{1}{m}$, the optimal mechanism given by~\eqref{remark:general_case_solving_rv1} illustrates that given $Y$ only one bit of additional information about $X$ is needed (namely, whether or not $X = Y$) in order obtain the optimal privacy-utility tradeoff for the full data scenario.

\begin{theorem} \label{thm:y_available}
\textbf{(Output Perturbation Privacy-Utility Tradeoff for the Symmetric Pair Distribution)}
With mutual information as the privacy measure, $J(X;Z) = I(X;Z)$, and probability of error as the distortion measure, $D(P_{Y,Z}) = \Pr(Y \neq Z)$, if the data model is $(X,Y) \sim SP(m,p)$, then
the optimal privacy-utility tradeoff for the output perturbation scenario in~\eqref{eqn:OP_opt_problem} is given by
\begin{align}\label{eqn:SP_OP_region}
\pi_\text{OP}(\delta) =
\begin{cases}
r_m \left( p + \delta \left(1 - \frac{pm}{m-1} \right) \right), & \text{if } \delta < 1-\frac{1}{m},\\
0, & \text{otherwise}.
\end{cases}
\end{align}
The optimal mechanism is given by $Z := Y + N \mod m$, where $N \in \{0, \ldots, m-1\}$ is independent of $(X,Y)$ with the distribution
\begin{align}\label{eqn:OP_opt_mech}
P_{N}(n) = 
\begin{cases}
1-t, & \text{if } n = 0 \\
\frac{t}{m-1}, & \text{otherwise},
\end{cases}
\end{align}
where $t := \min(\delta, 1-\frac{1}{m})$.
\end{theorem}

\noindent The proof of Theorem~\ref{thm:y_available}
is presented in Appendix~\ref{app:proof_y_available}.

For the output perturbation scenario, the optimal mechanism given in Theorem~\ref{thm:y_available} simply adds noise (see~\eqref{eqn:OP_opt_mech}) that results in a probability of error $\Pr(Y \neq Z)$ equal to the distortion budget $\delta$ (when it is less than $1-\frac{1}{m}$).
Note that this mechanism does not depend on the parameter $p$, and hence tolerates some statistical uncertainty regarding $(X,Y)$.

\begin{theorem}\label{thm:x_available}
\textbf{(Inference Privacy-Utility Tradeoff for the Symmetric Pair Distribution)}
With mutual information as the privacy measure, $J(X;Z) = I(X;Z)$, and probability of error as the distortion measure, $D(P_{Y,Z}) = \Pr(Y \neq Z)$, if the data model is $(X,Y) \sim SP(m,p)$, then
the optimal privacy-utility tradeoff for the inference scenario in~\eqref{eqn:INF_opt_problem} is given by
\begin{align}\label{eqn:SP_INF_region}
\pi_\text{INF}(\delta) =
\begin{cases}
r_m(t), & \text{if }
\delta < 1-\frac{1}{m} \text{ and } p \notin(\delta,h), \\
\infty, & \text{if } 
\delta < 1-\frac{1}{m} \text{ and } p \in(\delta,h), \\
0, & \text{if } \delta \geq 1-\frac{1}{m},
\end{cases}
\end{align}
where $h := (m-1)(1-\delta)$ and
\[
t := \frac{\delta-p}{1-\frac{pm}{m-1}}.
\]
\end{theorem}

\begin{remark}\textbf{(Tradeoff Plots)} In Figure~\ref{fig:trd1}, we plot the optimal privacy-utility tradeoff curves under the full data, output perturbation, and inference scenarios, for the symmetric pair data model with alphabet size $m = 10$ and cross-over parameter $p = 0.4$.

\begin{figure}[t]
   \centering
   \includegraphics[width=0.50\textwidth]{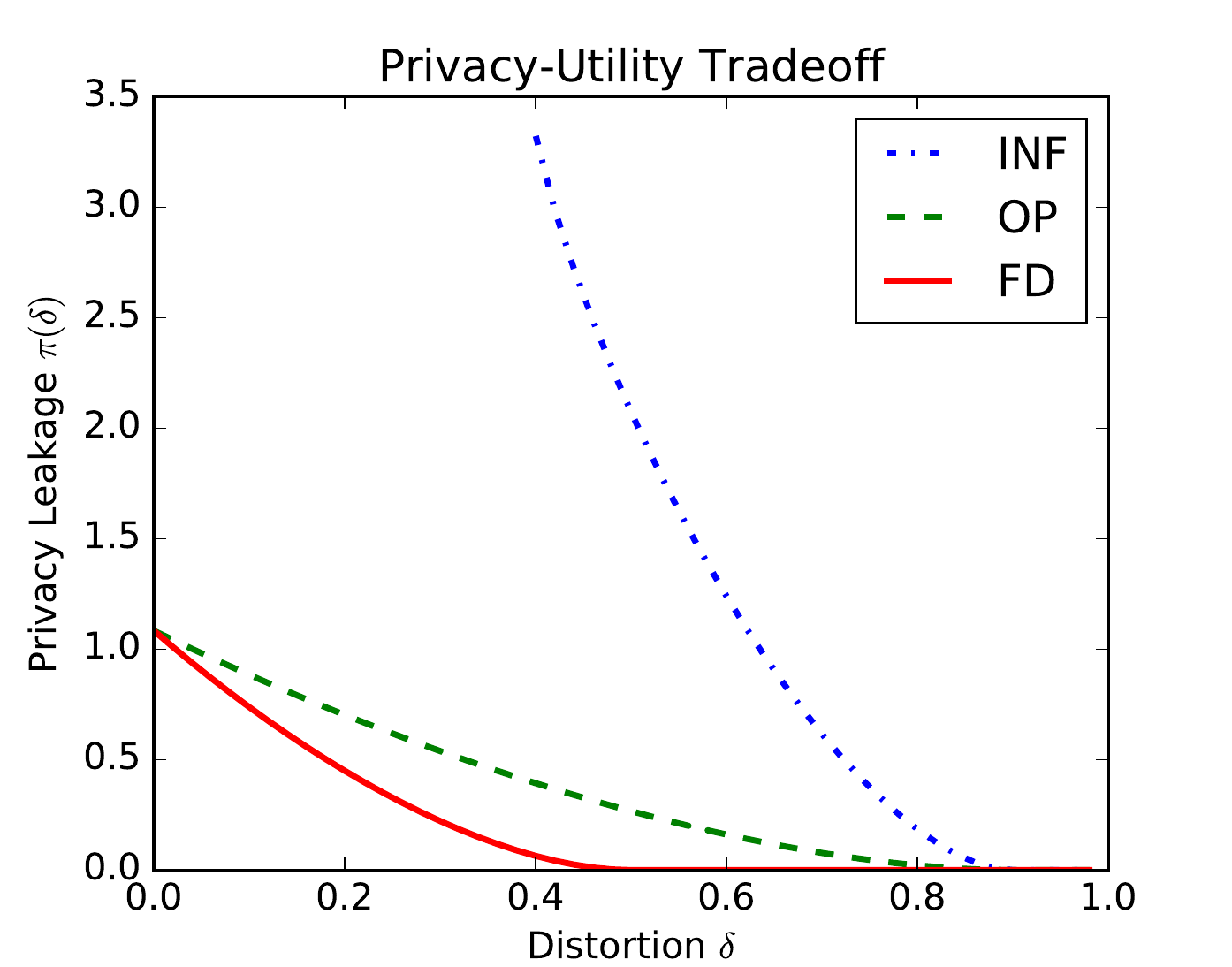}
   \caption{Optimal privacy-utility tradeoff curves under the inference (INF), output perturbation (OP), and full data (FD) scenarios, for the symmetric pair data model with alphabet size $m = 10$ and cross-over probability $p = 0.4$.}
   \label{fig:trd1}
 \end{figure}
\end{remark}

%% file: conclusion.tex
\section{Conclusion}
In this paper, we formulated the privacy-utility tradeoff problem where the data release mechanism has limited access to the entire data composed of useful and sensitive parts.
Based on this information theoretic formulation, we compared the privacy-utility tradeoff regions attained by full data, output perturbation, and inference mechanisms, which have access to the entire data, only useful data, and only sensitive data, respectively.

We first observed that the full data mechanism provides the best privacy-utility tradeoff and then showed that the output perturbation mechanism provides a better privacy-utility tradeoff than the inference mechanism.
We showed that if the common and mutual information between useful and sensitive data are identical, then the full data mechanism simplifies to the output perturbation mechanism.
Conversely, we showed that if the common information is not equal to mutual information, then the tradeoff region achieved by full data mechanism is strictly larger than the one achieved by the output perturbation mechanism.

Throughout the paper, we allowed for a general distortion measure, and
a general privacy measure that satisfies certain conditions that any
reasonable measure of privacy should satisfy. In particular, the measure
does not have to be symmetric and need not satisfy both the
inequalities that are usually implied by the data processing
inequality for a symmetric measure. In this context, the linkage
inequality was identified as the key property that is required for our
main results to hold. It was shown that the Sibson mutual information
of order infinity and the information privacy measures satisfy both the
post-processing and linkage inequalities, but the maximal information
leakage and differential privacy measures can violate the linkage
inequality. The philosophical implications of this for differential
privacy were also highlighted through a carefully
constructed analytical example.

%% file: metricProofs.tex
\section{Proofs of Section~\ref{sec:privacy} Results}
\label{app:metricProofs}

\subsection{Proof of Proposition~\ref{prop:MaxIPineq}}
\label{app:MaxIPineq}

For $X \to Z_1 \to Z_2$ that form a Markov chain, we have that
\begin{align*}
\min_{z_1} H(X | Z_1 = z_1) &= \min_{z_1, z_2} H(X | Z_1 = z_1, Z_2 = z_2) \\
&\leq \min_{z_2} H(X | Z_1, Z_2 = z_2) \\
&\leq \min_{z_2} H(X | Z_2 = z_2),
\end{align*}
and thus, $I^*(X; Z_1) \geq I^*(X; Z_2)$, establishing the
post-processing inequality.

Violation of the linkage inequality is shown by considering the
counter-example where $X_2$ is ternary with $P_{X_2}(0) = 1/2$ and
$P_{X_2}(1) = P_{X_2}(2) = 1/4$, $X_1$ is binary with $X_1 = 0$ if and
only if $X_2 = 0$, and the release $Z = X_1$.  For this example, $X_2
\to X_1 \to Z$ is a Markov chain, $I^*(X_1; Z) = H(X_1) = 1$, and
$I^*(X_2; Z) = H(X_2) = 1.5$, since $H(X_1 | Z = 0) = H(X_2 | Z = 0) =
0$.  Hence, $I^*(X_2; Z) > I^*(X_1; Z)$ and the linkage inequality
does not hold.

\subsection{Proof of Lemma~\ref{lem:EpsInfoPrivIneq}}\label{app:EpsInfoPrivIneq}

Due to symmetry, it suffices to show only the post-processing
inequality.  For $X \to Z_1 \to Z_2$ that form a Markov chain, we have
that
\begin{align*}
IP (X;Z_2) &= \max_{x,z_2} \left | \ln \frac{P_{X|Z_2}(x|z_2)}{P_X(x)} \right | \\
&= \max_{x,z_2} \left | \ln \sum_{z_1} \frac{P_{X|Z_1}(x|z_1)
  P_{Z_1|Z_2}(z_1|z_2) }{P_X(x)} \right | \\
&= \max_{x,z_2} \left | \ln E_{Z_1} \left[
  \frac{P_{X|Z_1}(x|Z_1)}{P_X(x)} \bigg| Z_2 = z_2 \right] \right | \\
&\leq \max_{x,z_1} \left | \ln \frac{P_{X|Z_1}(x|z_1)}{P_X(x)} \right
| = IP (X;Z_1),
\end{align*}
where each maximization is over the supports of the respective
marginal distributions, and the inequality follows since the
absolute-log of the expectation is bounded by the maximum of the
absolute-log over the support.

\subsection{Proof of Corollary~\ref{cor:InfoPrivEquiv}}
\label{app:InfoPrivEquiv}

From~\eqref{eqn:epsInfoPriv} it follows that
\begin{equation*}
IP(X;Z) \leq \max_{\substack{\mathcal{A}\subseteq \mathcal{X},
    \mathcal{B} \subseteq \mathcal{Z}:\\ \Pr(X \in \mathcal{A}), \Pr(Z \in
    \mathcal{B}) > 0}} \left | \ln \frac{\Pr(X\in \mathcal{A}, Z\in
  \mathcal{B})}{\Pr(X\in \mathcal{A}) \Pr(Z \in \mathcal{B})} \right |.
\end{equation*}
To demonstrate the reverse inequality, we first observe that
\begin{align*}
& \max_{\substack{\mathcal{A}\subseteq \mathcal{X}, \mathcal{B}
      \subseteq \mathcal{Z}:\\ \Pr(X \in \mathcal{A}), \Pr(Z \in
      \mathcal{B}) > 0}} \left | \ln \frac{\Pr(X\in \mathcal{A}, Z\in
    \mathcal{B})}{\Pr(X\in \mathcal{A}) \Pr(Z \in \mathcal{B})} \right |
  = \\
& \max_{\substack{\mathcal{A}\subseteq \mathcal{X}, \mathcal{B}
      \subseteq \mathcal{Z}:\\ \Pr(X \in \mathcal{A}), \Pr(Z \in
      \mathcal{B}) > 0}} IP(1(X\in \mathcal{A});1(Z \in \mathcal{B})).
\end{align*}
Next note that $ 1(X \in \mathcal{A}) \to X \to Z \to 1(Z \in
\mathcal{B})$ forms a Markov chain for any choice of
$\mathcal{A}\subseteq \mathcal{X}, \mathcal{B} \subseteq \mathcal{Z}$
such that $\Pr(X \in \mathcal{A}), \Pr(Z \in \mathcal{B}) > 0$. From
Lemma~\ref{lem:EpsInfoPrivIneq} it follows that $IP(1(X \in
\mathcal{A});1(Z\in \mathcal{B}))$ cannot be larger than $IP(X;Z)$
(post-processing and linkage inequalities) for any valid choice of
$\mathcal{A,B}$. Thus,
\begin{align*}
\max_{\substack{\mathcal{A}\subseteq \mathcal{X}, \mathcal{B}
    \subseteq \mathcal{Z}:\\ \Pr(X \in \mathcal{A}), \Pr(Z \in
    \mathcal{B}) > 0}} IP(1(X\in \mathcal{A});1(Z \in \mathcal{B}))
\leq IP(X;Z)
\end{align*}
and the result follows.

\subsection{Proof of Lemma~\ref{lem:DPmetric}}\label{app:DPmetric}

From the definition its follows that a release mechanism is $\epsilon$-differentially private
if, and only if, for all $x_1,x_2
\in \mathcal{X}$ with $d_H(x_1,x_2)=1$, and all $\mathcal{B} \subseteq
\mathcal{Z}$,
\[
\left | \ln \frac{\Pr(Z\in \mathcal{B}|x_1)}{\Pr(Z\in \mathcal{B}|x_2)}
\right | \leq \epsilon
\]
Thus if a release mechanism is $\epsilon$-differentially private, then
\begin{equation} \label{eqn:epsDiffPriv}
DP(X;Z) := \max_{\substack{x_1,x_2\in \mathcal{X}, \mathcal{B} \subseteq \mathcal{Z}:\\ d_H(x_1,x_2)=1}}
\left | \ln \frac{\Pr(Z\in \mathcal{B}|x_1)}{\Pr(Z\in \mathcal{B}|x_2)}
\right | \leq \epsilon.
\end{equation}
Since
\[
\frac{\Pr(Z\in \mathcal{B}|X=x_1)}{\Pr(Z\in \mathcal{B}|X=x_2)} =
\frac{\sum_{z\in \mathcal{B}} P_{Z|X}(z|x_1)}{\sum_{z\in \mathcal{B}}
  P_{Z|X}(z|x_2)} \leq \max_{z\in \mathcal{B}}
\frac{P_{Z|X}(z|x_1)}{P_{Z|X}(z|x_2)},
\]
it follows that reducing the scope of maximization in
\eqref{eqn:epsDiffPriv} from subsets $\mathcal{B} \subseteq
\mathcal{Z}$ to singletons $z \in \mathcal{Z}$ will not decrease the
maximum value, i.e.,

\subsection{Proof of Proposition~\ref{prop:DPexample}}
\label{app:DPexample}

\begin{figure}[!t]
   \centering
   \includegraphics[width=0.48\textwidth]{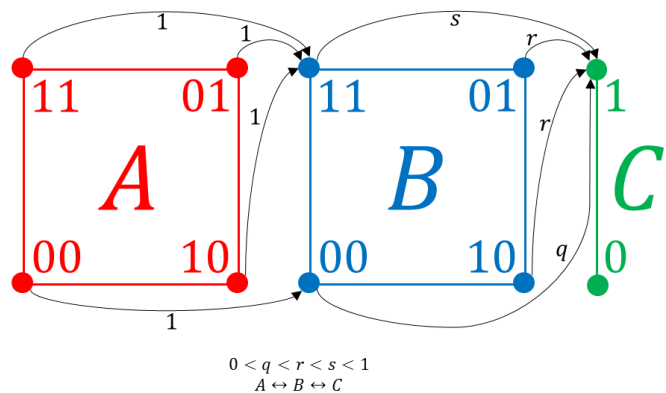}
\vglue -10ex
   \caption{An example which demonstrates that DP can violate the
     linkage inequality. Here, $A$ and $B$ are databases taking values
     in $\{0,1\}^2$ and $C\in \{0,1\}$ is the data release. The
     colored edges indicate databases that differ in exactly one
     element. By construction, $A\to B \to C$ forms a Markov chain and
     yet $DP(A;C) > DP(B;C)$.}
   \label{fig:DPviolatesDPI}
\end{figure}

We will construct $A \to B \to C$ such that $DP(A;C) >
DP(B;C)$.  Let databases $A, B \in \{0,1\}^2$ and the release $C \in
\{0,1\}$. The database $B := (B_1,B_2)$ is a deterministic function
of the database $A := (A_1,A_2)$. Specifically, $B_1 = B_2 = A_1 \vee
A_2$. The release $C$ is produced by the mechanism $P_{C|B}$, given by
\[
P_{C|B}(1 | b) =\begin{cases}
q, & \text{if }b=(0,0),\\
s, & \text{if }b=(1,1),\\
r, & \text{otherwise},
\end{cases}
\]
with $0 < q < r < s < 1$.
The construction of $(A,B,C)$ is summarized in
Fig.~\ref{fig:DPviolatesDPI} where the solid circles indicate databases and
the colored edges join databases that are at Hamming distance one
from each other.
Since $0 < q < r < s < 1$, 
\[
1 < \max\left(\frac{s}{r},\frac{r}{q}\right) < \frac{s}{q}
\]
If we define $\bar{t}:= (1-t)$ for convenience, then $0 < \bar{s} <
\bar{r} < \bar{q} < 1$ so that
\[
1 < \max\left(\frac{\bar{r}}{\bar{s}},\frac{\bar{q}}{\bar{r}}\right) <
\frac{\bar{q}}{\bar{s}}.
\]
Thus,
\begin{align*}
0 = \ln 1 < DP(B;C) &=
\max\left(\ln\frac{s}{r},\ln\frac{r}{q},\ln\frac{\bar{q}}{\bar{r}},\ln\frac{\bar{r}}{\bar{s}}\right)
\\ &< \max\left(\ln\frac{s}{q},\ln\frac{\bar{q}}{\bar{s}}\right)
= DP(A;C).
\end{align*}

%% file: orderProof.tex
\section{Properties of Common Information} \label{app:CommonInfoProps}

The {\em graphical representation} of $P_{X,Y}$ is the bipartite graph
with an edge between $x \in \mathcal{X}$ and $y \in \mathcal{Y}$ if
and only if $P_{X,Y}(x,y) > 0$.  The {\em common part} $U$ of two random
variables $(X,Y)$ is defined as the (unique) label of the
connected component of the graphical representation of $P_{X,Y}$ in
which $(X,Y)$ falls. Note that $U$ is a deterministic function
of $X$ alone and also a deterministic function of $Y$ alone.

The G\'acs-K\"orner common information of two random variables $(X,Y)$
is given by entropy of the common part, that is, $C(X;Y) := H(U)$, and has
the operational significance of being the maximum number
of common bits per symbol that can be independently extracted from $X$
and $Y$~\cite{GacsKorn-73-CommonInfo}.
In general, $C(X;Y) \leq I(X;Y)$, with equality if and only if $X \to U \to Y$ forms a Markov chain~\cite{AhlsKorn-74-CommonInfo}.
Since our results are only concerned with whether $C(X;Y) = I(X;Y)$, our theorem statements are unchanged if we use instead the Wyner notion of common information (see~\cite{Wyner-75-CommonInfo}), since it is also equal to mutual information if and only if $X \to U \to Y$ forms a Markov chain~\cite{AhlsKorn-74-CommonInfo}.

We give the following lemma which aids our proof of Theorem~\ref{theorem:OPConverse} in Appendix~\ref{app:OPConverseProof}.

\begin{lemma} \label{lemma:CommonPartProperty}
If $C(X;Y) \neq I(X;Y)$, then there exist $x_0, x_1 \in \mathcal{X}$ and $y_0, y_1 \in \mathcal{Y}$, such that $y_0 \neq y_1$, $P_{X,Y}(x_0,y_0) > 0$, $P_{X,Y}(x_0,y_1) > 0$, and $P_{X|Y}(x_1|y_0) \neq P_{X|Y}(x_1|y_1)$.
\end{lemma}

\begin{proof}
We will prove this lemma by showing the contrapositive, that is, if there does not exist $x_0, x_1 \in \mathcal{X}$ and $y_0, y_1 \in \mathcal{Y}$ satisfying the conditions stated in the lemma, then $C(X;Y) = I(X;Y)$.
First, note that if for all $x_0 \in \mathcal{X}$ and $y_0, y_1 \in \mathcal{Y}$, either $y_0 = y_1$, $P_{X,Y}(x_0,y_0) = 0$, or $P_{X,Y}(x_0,y_1) = 0$, then $Y$ is a deterministic function of $X$, which would result in $C(X;Y) = I(X;Y)$.
Thus, we are left with showing that for all $x_0 \in \mathcal{X}$ and $y_0, y_1 \in \mathcal{Y}$, with $y_0 \neq y_1$, $P_{X,Y}(x_0,y_0) > 0$, and $P_{X,Y}(x_0,y_1) > 0$, if we also have that for all $x_1 \in \mathcal{X}$, $P_{X|Y}(x_1|y_0) = P_{X|Y}(x_1|y_1)$, then $C(X;Y) = I(X;Y)$.
This follows since these conditions would imply that for the common part $U$ of $(X,Y)$, $X \to U \to Y$ forms a Markov chain.
\end{proof}

\section{Proof of Theorem~\ref{theorem:order}} \label{app:OrderProof}

It is sufficient to show that for any mechanism $P_{Z|X}$ that is a feasible solution in the inference optimization of~\eqref{eqn:INF_opt_problem}, there is a corresponding mechanism $P_{Z'|Y}$ for the output perturbation optimization of~\eqref{eqn:OP_opt_problem} that achieves the same distortion and only lesser or equal privacy-leakage.

Let $P_{Z|X}$ be a mechanism in the feasible region of the inference optimization problem of~\eqref{eqn:INF_opt_problem}.
Define the corresponding mechanism for the output perturbation optimization of~\eqref{eqn:OP_opt_problem} by
\[
P_{Z'|Y}(z|y) := \sum_{x \in \mathcal{X}} P_{Z|X}(z|x) P_{X|Y}(x|y).
\]
Let $(X, Y, Z, Z') \sim P_{X,Y} P_{Z|X} P_{Z'|Y}$.
Note that by construction, $(Y,Z)$ and $(Y,Z')$ have the same distribution $P_{Y} P_{Z'|Y}$.
Thus, both mechanisms achieve the same distortion $D(P_{Y} P_{Z'|Y})$ and $J(Y;Z) = J(Y;Z')$.
Further, by construction, $Y \to X \to Z$ and $X \to Y \to Z'$ form Markov chains.
Thus, by the linkage inequality,
\[
J(X;Z') \leq J(Y;Z') = J(Y;Z) \leq J(X;Z),
\]
showing that the output perturbation mechanism has only lesser or equal privacy-leakage.

%% file: OPoptProof.tex
\section{Proof of Theorem~\ref{theorem:OPoptimality}} \label{app:OPoptimalityProof}

Since $\pi_{\text{FD}}(\delta) \leq \pi_{\text{OP}}(\delta)$ is immediate, we only need to show that
$\pi_{\text{OP}}(\delta) \leq \pi_{\text{FD}}(\delta)$.
It is sufficient to show that for any mechanism $P_{Z|X,Y}$ that is a feasible solution in the full data optimization of~\eqref{eqn:FD_opt_problem}, there is a corresponding mechanism $P_{Z'|Y}$ for the output perturbation optimization of~\eqref{eqn:OP_opt_problem} that achieves the same distortion and only lesser or equal privacy-leakage.

Let $P_{Z|X,Y}$ be a mechanism in the feasible region of the full data optimization problem of~\eqref{eqn:FD_opt_problem}.
Define the corresponding mechanism for the output perturbation optimization of~\eqref{eqn:OP_opt_problem} by
\[
P_{Z'|Y}(z|y) := \sum_{x \in \mathcal{X}} P_{Z|X,Y}(z|x,y) P_{X|Y}(x|y).
\]
Let $(X, Y, Z, Z') \sim P_{X,Y} P_{Z|X,Y} P_{Z'|Y}$, and let $U$ be the common part of $(X,Y)$, where, by construction, $U$ is a deterministic function of either $X$ alone or $Y$ alone.
Since $C(X;Y) = I(X;Y)$, we have that $X \to U \to Y$ forms
a Markov chain, i.e., $I(X;Y|U) = 0$.
By construction, $X \to Y \to Z'$ also forms a Markov chain, and hence $I(X;Z'|U,Y) = I(X;Z'|Y) = 0$, since $U$ is deterministic function of $Y$.
Given these two Markov chains, we have
\begin{align*}
0 &= I(X; Y|U) + I(X; Z'|U,Y) \\
  &= I(X; Y,Z'|U) \\
  &= I(X; Z'|U) + I(X; Y|U,Z') \\
  &\geq I(X; Z'|U),
\end{align*}
and hence $I(X; Z'|U) = 0$, i.e., $X \to U \to Z'$ also forms a Markov chain.
Continuing, we can show the desired privacy-leakage inequality as follows
\begin{align*}
J(X;Z') &\leq J(X,U;Z') \\
&\leq J(U;Z') \\
&= J(U;Z) \\
&\leq J(X,U;Z) \\
&\leq J(X;Z),
\end{align*}
where the equality holds since by construction $P_{Y,Z} = P_{Y,Z'}$
(and hence $P_{U,Z} = P_{U,Z'}$), and the four inequalities follow, respectively, by applying the linkage inequality to the following Markov chains:
\begin{itemize}
\item $X \to (X,U) \to Z'$, since $X$ is a function of $(X,U)$.
\item $(X,U) \to U \to Z'$, since $U$ is a function of $X$, and since $X \to U \to Z'$ forms a Markov chain as shown above.
\item $U \to (X,U) \to Z$, since $U$ is a function of $(X,U)$.
\item $(X,U) \to X \to Z$, since $(X,U)$ is a function of $X$.
\end{itemize}

%% file: NecessaryCondProof.tex
\section{Proof of Theorem~\ref{theorem:OPConverse}} \label{app:OPConverseProof}

We will show the following result, which is key to the proof.

\begin{lemma} \label{lemma:NecessaryCond}
If $C(X;Y) \neq I(X;Y)$ then there exist random variables $Z$ and $Z'$ with $P_{Y,Z} = P_{Y,Z'}$, such that $X \to Y \to Z'$ forms a Markov chain,
$I(X;Z) = 0$, and $I(X;Z') > 0$.
\end{lemma}

The proof of Theorem~\ref{theorem:OPConverse} then follows by defining the distortion measure $D(P_{Y,Z})$ to equal $1$
for the particular choice of $P_{Y,Z'}$ in
Lemma~\ref{lemma:NecessaryCond} and to equal $2$ otherwise, and choosing
$\delta = 1$.
This choice for the distortion measure and distortion level restricts the feasible output perturbation mechanism to only $P_{Z'|Y}$, which by Lemma~\ref{lemma:NecessaryCond} results in $\pi_{OP}(\delta) > 0$ since $J(X;Z') > 0$ (since $I(X;Z') > 0$).
However, the proof of Lemma~\ref{lemma:NecessaryCond} (see below) also ensures the existence of $Z$ produced by a full data mechanism $P_{Z|X,Y}$ that results in $\pi_{FD}(\delta) = 0$ since $I(X;Z) = 0$.

Using the symbols $(x_0,x_1,y_0,y_1)$ shown to exist by Lemma~\ref{lemma:CommonPartProperty}, we can prove Lemma~\ref{lemma:NecessaryCond} by constructing a binary $Z$ with alphabet $\mathcal{Z} = \{0,1\}$ as follows.
Choose any $s \in (0,1)$ and
any $t \in \left(0, \min\{ s'/P_{Y|X}(y_1|x_0) , s/P_{Y|X}(y_0|x_0) \} \right)$, where $s' :=
(1-s)$.
Define $Z$ with $(X,Y,Z) \sim P_{X,Y}P_{Z|X,Y}$, where
\begin{align*}
P_{Z|X,Y}(0|x,y) := \begin{cases}
s + t P_{Y|X}(y_1|x_0), & \text{if } (x,y) = (x_0,y_0), \\
s - t P_{Y|X}(y_0|x_0), & \text{if } (x,y) = (x_0,y_1), \\
s, & \text{otherwise}.
\end{cases}
\end{align*}
The choice of $s$ and $t$ ensures that $P_{Z|X,Y}(0|x,y) \in (0,1)$ for all
$(x,y) \in \mathcal{X} \times \mathcal{Y}$.
This construction of $P_{Z|X,Y}$ makes $Z$ independent of $X$, since for all $x \in \mathcal{X}$ in the support of $P_X$,
\[
P_{Z|X}(0|x) = \sum_{y \in \mathcal{Y}} P_{Z|X,Y}(0|x,y) P_{Y|X}(y|x) = s.
\]
With the above construction, we have
\begin{align*}
P_{Z|Y}(0|y) &= \sum_{x \in \mathcal{X}} P_{Z|X,Y}(0|x,y) P_{X|Y}(x|y) \\
&= \begin{cases}
s + t P_{Y|X}(y_1|x_0) P_{X|Y}(x_0|y_0), & \text{if } y = y_0, \\
s - t P_{Y|X}(y_0|x_0) P_{X|Y}(x_0|y_1), & \text{if } y = y_1, \\
s, & \text{otherwise}.
\end{cases}
\end{align*}

Next, we construct binary $Z'$ such that $X \to Y \to Z'$ forms a Markov chain, with $(X, Y, Z') \sim P_{X,Y} P_{Z'|Y}$, where we set $P_{Z'|Y} := P_{Z|Y}$.
Then, consider
\begin{align*}
P_{Z'|X}(0|x) &= \sum_{y \in \mathcal{Y}} P_{Z'|Y}(0|y) P_{Y|X}(y|x) \\
&= \sum_{y \in \mathcal{Y}} P_{Z|Y}(0|y) P_{Y|X}(y|x) \\
&= s + t P_{Y|X}(y_1|x_0) P_{X|Y}(x_0|y_0) P_{Y|X}(y_0|x) \\
& \quad \quad - t P_{Y|X}(y_0|x_0) P_{X|Y}(x_0|y_1) P_{Y|X}(y_1|x) \\
&= s + t P_X(x_0) P_{Y|X}(y_0|x_0) P_{Y|X}(y_1|x_0) \\
& \quad \quad \times [ P_{X|Y}(x|y_0) - P_{X|Y}(x|y_1)]/P_X(x).
\end{align*}
Finally, we show that $P_{Z'|X}(0|x)$ is not constant for all $x \in \mathcal{X}$ in the support of $P_X$,
which implies that $Z'$ is not independent of $X$, i.e., $I(X;Z') > 0$.
This can be proved by contradiction, by supposing that $P_{Z'|X}(0|x)$ is constant for all $x \in \mathcal{X}$ in the support of $P_X$.
Then, for all $x \in \mathcal{X}$,
\[
P_{X|Y}(x|y_0) - P_{X|Y}(x|y_1) = c P_X(x),
\]
for some constant $c$. By summing over all $x \in \mathcal{X}$, we have that $c = 0$.
This would imply that $P_{X|Y}(x|y_0) = P_{X|Y}(x|y_1)$ for all $x \in \mathcal{X}$, contradicting the existence of $x_1 \in \mathcal{X}$ given by
Lemma~\ref{lemma:CommonPartProperty} for the choice of $y_0$ and $y_1$.

%% file: lemmaproofs.tex
\section{Some Useful Lemmas}
In this section, we provide a set of lemmas that we use to prove the results presented in Section~\ref{sec:example}.

\begin{lemma}\label{lemma:error_relation}
Let $X,Y$, and $Z$ be discrete random variables, with $X, Y \in \{0,\dots, m-1\}$.
If $(X,Y) \sim SP(m, p)$, then
\begin{align*}
&\Pr(Y \neq Z) - \Pr(X \neq Z) \\
&\quad =\frac{p}{m(m-1)} \mathop{\sum_{x,y}}_{x \neq y} \left[ P_{Z|X,Y}(x|x,y)-P_{Z|X,Y}(y|x,y) \right]
\end{align*}
\end{lemma}

\begin{proof}
We can expand $\Pr(Y\neq Z)$ as 
\begin{align*}
\Pr(Y \neq Z) &= 1 - \Pr(Y=Z) \\
& = 1 - \sum_{x,y} P_{X,Y}(x,y)P_{Z|X,Y}(y|x,y) \\
& = 1 - \sum_x P_{X,Y}(x,x)P_{Z|X,Y}(x|x,x) \\
& \quad - \sum_{x \neq y}P_{X,Y}(x,y) P_{Z|X,Y}(y|x,y) \\
& = 1 - \frac{1-p}{m} \sum_x P_{Z|X,Y}(x|x,x) \\
& \quad - \frac{p}{m(m-1)} \sum_{x\neq y}P_{Z|X,Y}(y|x,y).
\end{align*}
Similarly, we have that
\begin{align*}
P(X \neq Z) &= 1 - \frac{1-p}{m} \sum_x P_{Z|X,Y}(x|x,x) \\
& \quad - \frac{p}{m(m-1)} \sum_{x\neq y} P_{Z|X,Y}(x|x,y).
\end{align*}
Subtracting these two expansions yields the lemma.
\end{proof}

\begin{lemma}\label{lemma:error_relation_markov}
Let $X,Y$, and $Z$ be discrete random variables, with $X, Y \in \{0,\dots, m-1\}$.
If $(X,Y) \sim SP(m, p)$ and $Y\to X\to Z$ forms a Markov chain, then
\begin{equation*}
\Pr(Y\neq Z) =  p+\Pr(X\neq Z) \left( 1 - \frac{pm}{m-1} \right).
\end{equation*}
\end{lemma}
\begin{proof}
We have that
\begin{align*}
&\Pr(Y\neq Z) - \Pr(X\neq Z) \\
&\quad \stackrel{(a)}{=} \frac{p}{m(m-1)}\mathop{\sum_{x,y} }_{x\neq y}P_{Z|X,Y}(x|x,y)-P_{Z|X,Y}(y|x,y) \\
&\quad \stackrel{(b)}{=} \frac{p}{m(m-1)}\mathop{\sum_{x,y} }_{x\neq y}P_{Z|X}(x|x)-P_{Z|X}(y|x) \\
&\quad \stackrel{(c)}{=} \frac{p}{(m-1)}\mathop{\sum_{x,y} }_{x\neq y}P_{X,Z}(x,x)-P_{X,Z}(x,y) \\
&\quad=\frac{p}{(m-1)}((m-1)\Pr(X=Z)-P(X\neq Z)) \\
&\quad= p - \Pr(X\neq Z) \left( p + \frac{p}{m-1} \right),
\end{align*}
where (a) follows from Lemma~\ref{lemma:error_relation}, (b) since $Y\to X\to Z$ forms a Markov chain, and (c) since $X$ is uniform over $\{0,\dots, m-1\}$.
Rearranging terms yields the lemma.
\end{proof}

\begin{lemma}\label{lemma_fano1}
Let $X$ be uniformly distributed on $\{0,\dots, m-1\}$ and define
\begin{align*}
f(\epsilon ):= & \inf_{P_{Z|X}} I(X;Z) \\
& \ \ \text{s.t.} \ \Pr(X\neq Z)\leq \epsilon,
\end{align*}
and
\begin{align*}
g(\epsilon) :=
\begin{cases}
r_m(\epsilon), & \text{if } \epsilon \leq 1-\frac{1}{m}\\
0, & \text{otherwise},
\end{cases}
\end{align*}
for $\epsilon\in [0,\infty)$.
Then, $f(\epsilon)$ = $g(\epsilon)$ for any $\epsilon \in[0,\infty)$, with the mechanism $P_{Z|X}$ solving the optimization problem given by
\begin{align}\label{app:test_channel}
P_{Z|X}(z|x) = 
\begin{cases}
{1-t}, \text{  if } z=x\\
\frac{t}{m-1}, \text{  otherwise,}
\end{cases}
\end{align}
where $t := \min (1-\frac{1}{m}, \epsilon)$ and $z \in \{0,\dots, m-1\}$.
\end{lemma}

\begin{proof}
We immediately have that $f(\epsilon) = g(\epsilon) = 0$ for any $\epsilon > 1-\frac{1}{m}$, since for $Z$ independent of $X$ and uniformly distributed over $\{0, \ldots, m-1\}$, which is consistent with~\eqref{app:test_channel}, we have that $\Pr(X \neq Z) = 1-\frac{1}{m}$ and $I(X;Z) = 0$.
Thus, for the rest of the proof, we assume that $\epsilon \leq 1 -\frac{1}{m}$.

We first show that $f(\epsilon)\geq g(\epsilon)$, using a lower bound on $I(X;Z)$,
\begin{align}
I(X;Z) &= \log m - H(X|Z) \nonumber\\
&\geq r_m(\Pr(X\neq Z)), \label{eqn:fano-bound}
\end{align}
which follows from Fano's inequality and definition of $r_m$ from Lemma~\ref{lemma:sp_mut_computation}.
Thus, for $\epsilon \leq 1 -\frac{1}{m}$,
\begin{align*}
f(\epsilon)&\geq \inf_{P_{Z|X}} r_m(\Pr(X\neq Z)) \quad \text{s.t. }  \Pr(X\neq Z)\leq \epsilon \\
& = r_m(\epsilon) =: g(\epsilon),
\end{align*}
since $r_m(\epsilon)$ is strictly decreasing over $[0, 1-\frac{1}{m}]$.

We next show that $f(\epsilon) \leq g(\epsilon)$ for $P_{Z|X}$ given by~\eqref{app:test_channel}.
Note that $(X,Z)\sim SP(m,t)$, and hence the conditional probability $P_{Z|X}$ is in the feasible region of the optimization problem since $\Pr(X\neq Z) = t = \epsilon$.
Consequently, we have $f(\epsilon) \leq I(X;Z) = r_m(t) = g(\epsilon)$, where the first equality follows from Lemma~\ref{lemma:sp_mut_computation} and the second since $t = \epsilon \leq 1-\frac{1}{m}$.
\end{proof}

\begin{lemma}\label{fano_lower_bound2}
Let $X$ be uniformly distributed on $\{0,\dots, m-1\}$ and define
\begin{align*}
f^*(\epsilon ):= & \inf_{P_{Z|X}} I(X;Z) \\
& \ \ \text{s.t.} \ \Pr(X\neq Z) \geq \epsilon,
\end{align*}
and
\begin{align*}
g^*(\epsilon) :=
\begin{cases}
r_m(\epsilon), & \text{if } \epsilon \geq 1-\frac{1}{m}\\
0, & \text{otherwise},
\end{cases}
\end{align*}
for $\epsilon\in [0,1]$. Then, $f^*(\epsilon)$ = $g^*(\epsilon)$ for any $\epsilon \in[0,1]$,
with the mechanism $P_{Z|X}$ solving the optimization problem given by~\eqref{app:test_channel} with $t := \max (1-\frac{1}{m}, \epsilon)$.
\end{lemma}

\begin{proof}
We immediately have that $f^*(\epsilon) = g^*(\epsilon) = 0$ for any $\epsilon < 1-\frac{1}{m}$, since for $Z$ independent of $X$ and uniformly distributed over $\{0, \ldots, m-1\}$, which is consistent with~\eqref{app:test_channel} with $t := \max (1-\frac{1}{m}, \epsilon)$, we have that $\Pr(X \neq Z) = 1-\frac{1}{m}$ and $I(X;Z) = 0$.
Thus, for the rest of the proof, we assume that $\epsilon \geq 1 -\frac{1}{m}$.

We first show that $f^*(\epsilon)\geq g^*(\epsilon)$, applying the lower bound of~\eqref{eqn:fano-bound} to yield
\begin{align*}
f^*(\epsilon) &\geq \inf_{P_{Z|X}} r_m(\Pr(X\neq Z)) \quad \text{s.t. }  \Pr(X\neq Z)\geq \epsilon \\
& = r_m(\epsilon) =: g^*(\epsilon),
\end{align*}
which follows since $r_m(\epsilon)$ is strictly increasing over $[1-\frac{1}{m}, 1]$.

We next show that $f^*(\epsilon) \leq g^*(\epsilon)$ for $P_{Z|X}$ given by~\eqref{app:test_channel} with $t := \max (1-\frac{1}{m}, \epsilon)$.
Note that $(X,Z)\sim SP(m,t)$, and hence the conditional probability $P_{Z|X}$ is in the feasible region of the optimization problem since $\Pr(X\neq Z) = t = \epsilon$.
Consequently, we have $f(\epsilon) \leq I(X;Z) = r_m(t) = g(\epsilon)$, where the first equality follows from Lemma~\ref{lemma:sp_mut_computation} and the second since $t = \epsilon \geq 1-\frac{1}{m}$.
\end{proof}

%% file: xyavailable.tex
\section{Proof of Theorem~\ref{thm:xy_available}}\label{app:proof_xy_available}

For convenience, we define
\begin{align*}
g_\text{FD}(\delta) :=
\begin{cases}
r_m(p + \delta), & \text{if } \delta \leq 1-\frac{1}{m} - p, \\
r_m(p - \delta), & \text{if } \delta \leq p - (1 - \frac{1}{m}), \\
0, & \text{otherwise}.
\end{cases}
\end{align*}
which is equal to the right-hand side of~\eqref{eqn:SP_FD_region}.
Since, for $\delta \geq 1 - \frac{1}{m}$, we immediately have $g_\text{FD}(\delta) = \pi_\text{FD}(\delta) = 0$, we will assume that $\delta < 1 - \frac{1}{m}$ for the rest of this proof.

We divide the proof into two cases: \textbf{(i)} $p \leq 1-\frac{1}{m}$ and \textbf{(ii)} $p > 1-\frac{1}{m}$. \\

\textbf{Case 1}: $p\leq 1-\frac{1}{m}$\\

We first show that $\pi_\text{FD}(\delta) \geq g_\text{FD}(\delta)$.
Due to Lemma~\ref{lemma:error_relation}, we have that
$\Pr(Y \neq Z) \leq \delta$ implies that
\begin{align*}
& \Pr(X\neq Z) \leq \delta \\
&\quad\quad +\frac{p}{m(m-1)}\mathop{\sum_{x,y} }_{x\neq y} \left[ P_{Z|X,Y}(y|x,y) - P_{Z|X,Y}(x|x,y) \right],\\
&\quad\leq \delta+ \frac{p}{m(m-1)}\mathop{\sum_{x,y} }_{x\neq y}1\\
&\quad\leq \delta + p,
\end{align*}
Thus, for any mechanism $P_{Z|X,Y}$ with $\Pr(Y \neq Z) \leq \delta $, we have that $\Pr(X \neq Z) \leq \delta +p$.
Then, we bound $\pi_\text{FD}(\delta)$ via 
\begin{align*}
\pi_\text{FD}(\delta) &\geq \inf_{P_{Z|X,Y}} I(X;Z) \quad \text{s.t. }  \Pr(X\neq Z)\leq \delta+p \\
&= \inf_{P_{Z|X}} I(X;Z) \quad \text{s.t. }  \Pr(X\neq Z)\leq \delta+p\\
&= g_\text{FD}(\delta),
\end{align*}
where the last equality follows from Lemma~\ref{lemma_fano1}.

We next show that $\pi_\text{FD}(\delta) \leq g_\text{FD}(\delta)$
via the mechanism given by~\eqref{remark:general_case_solving_rv1}, which is feasible since
\begin{align*}
\Pr(Y \neq Z) &= \Pr(Y \neq Z | X = Y) \Pr(X = Y)\\
&=\Pr(N \neq 0 | X = Y) (1-p)\\
&= t \leq \delta.
\end{align*}
Hence, we have that $\pi_\text{FD}(\delta) \leq I(X;Z)$.
For all $x \neq z$, we have that
\begin{align*}
P_{X,Z}(x,z) &= \sum_y P_{Z|X,Y}(z|x,y) P_{X,Y}(x,y) \\
&= P_{Z|X,Y}(z|x,x) P_{X,Y}(x,x) \\
&\quad + \sum_{y \neq x} P_{Z|X,Y}(z|x,y) P_{X,Y}(x,y) \\
&= \frac{t}{(1-p)(m-1)} \frac{1-p}{m} + P_{X,Y}(x,z) \\
&= \frac{t}{(m-1)m} + \frac{p}{(m-1)m} \\
&= \frac{t+p}{(m-1)m},
\end{align*}
which shows that $(X,Z)\sim SP(m, t+p)$.
Thus, by Lemma~\ref{lemma:sp_mut_computation}, $I(X;Z) = r_m(t+p)$.
Noting that $r_m(1-\frac{1}{m}) = 0$, we have $r_m(t+p) = g_\text{FD}(\delta)$ for $p \leq 1-\frac{1}{m}$.
Hence, $\pi_\text{FD}(\delta) \leq g_\text{FD}(\delta)$. \\

\textbf{Case 2}: $p > 1-\frac{1}{m}$\\

We first show that $\pi_\text{FD}(\delta) \geq g_\text{FD}(\delta)$.
Given $\Pr(Y \neq Z) \leq \delta$, we have that
\begin{align*}
\Pr(X\neq Z) + \delta &\geq \Pr(X\neq Z)+\Pr(Y\neq Z) \\
&\geq \Pr(\{X\neq Z\}\cup \{Y \neq Z\}) \\
&=1 - \Pr(X=Y=Z) \\
&\geq 1- \Pr(X=Y)\\
&= p.
\end{align*}
Thus, for any mechanism $P_{Z|X,Y}$ that satisfies $\Pr(Y\neq Z)\leq \delta $, we also have that $\Pr(X\neq Z) \geq p-\delta$.
Then, we can bound $\pi_\text{FD}(\delta)$ via
\begin{align*}
\pi_\text{FD}(\delta) &\geq \inf_{P_{Z|X,Y}} I(X;Z) \quad \text{s.t. }  \Pr(X \neq Z) \geq p-\delta\\
&= \inf_{P_{Z|X}} I(X;Z) \quad \text{s.t. } \Pr(X \neq Z) \geq p-\delta \\
&= g_\text{FD}(\delta),
\end{align*}
where the last equality follows from Lemma~\ref{fano_lower_bound2}.

We next show $\pi_\text{FD}(\delta) \leq g_\text{FD}(\delta)$, by considering the mechanism defined by
\begin{align*}
Z :=
\begin{cases}
Y,& \text{if } \theta'=1,\\
X,& \text{otherwise},
\end{cases}
\end{align*}
where $\theta'$ is a binary random variable that is independent of $(X,Y)$, with $\Pr(\theta' = 1) = t'/p$, where we define $t':= \max(p-\delta, 1-\frac{1}{m})$ for convenience.
Since
\begin{align*}
\Pr(Y\neq Z)&=\Pr(Y\neq Z | \theta'=0)P_{\theta'}(0)\\
&=\Pr(Y\neq X | \theta'=0)P_{\theta'}(0)\\
&=\Pr(Y\neq X)P_{\theta'}(0)\\
&=p-t'\\
&\leq \delta,
\end{align*}
we have that this mechanism is feasible.
Hence, we have $\pi_\text{FD}(\delta) \leq I(X;Z)$.
For all $x \neq z$, we have that
\begin{align*}
P_{X,Z}(x,z) &= \sum_y P_{Z|X,Y}(z|x,y) P_{X,Y}(x,y) \\
&= P_{Z|X,Y}(z|x,x) P_{X,Y}(x,x) \\
&\quad + \sum_{y\neq x} P_{Z|X,Y}(z|x,y) P_{X,Y}(x,y) \\
&= 0 + P_{Z|X,Y}(z|x,z) P_{X,Y}(x,z) + 0 \\
&= \frac{t'}{p} \frac{p}{(m-1)m} \\
&= \frac{t'}{(m-1)m},
\end{align*}
which shows that $(X,Z)\sim SP(m, t')$.
Thus, by Lemma~\ref{lemma:sp_mut_computation}, $I(X;Z)=r_m(t')$.
Noting that $r_m(1-\frac{1}{m}) = 0$, we have $r_m(t') = g_\text{FD}(\delta)$ for all $p \geq 1-\frac{1}{m}$.
Hence, $\pi_\text{FD}(\delta) \leq g_\text{FD}(\delta)$.

%% file: yavailable.tex
\section{Proof of Theorem~\ref{thm:y_available}}\label{app:proof_y_available}

For convenience, we define 
\begin{align*}
g_\text{OP}(\delta) :=
\begin{cases}
r_m \left( p + \delta \left(1 - \frac{pm}{m-1} \right) \right), & \text{if } \delta < 1-\frac{1}{m},\\
0, & \text{otherwise}.
\end{cases}
\end{align*}
which is equal to the right-hand side of~\eqref{eqn:SP_OP_region}.
Since, for $\delta \geq 1 - \frac{1}{m}$, we immediately have $g_\text{OP}(\delta) = \pi_\text{OP}(\delta) = 0$, we will assume that $\delta < 1 - \frac{1}{m}$ for the rest of this proof.

We first show that $\pi_\text{OP}(\delta) \geq g_\text{OP}(\delta)$.
Since $X \to Y \to Z$ forms a Markov chain for any output perturbation mechanism, we have from Lemma~\ref{lemma:error_relation_markov} that
\[
\Pr(X \neq Z) = p + \Pr(Y \neq Z) \left(1-\frac{pm}{m-1}\right).
\]
Let $\delta' := p+\delta\left(1-\frac{pm}{m-1}\right)$.
Note that when $p\leq 1-\frac{1}{m}$, the term $\left(1-\frac{pm}{m-1}\right) \geq 0$.
Hence, the constraint $\Pr(Y \neq Z) \leq \delta$ is equivalent to $\Pr(X \neq Z) \leq \delta'$, and $\delta' < 1-\frac{1}{m}$ since $\delta < 1-\frac{1}{m}$.
Thus, for $p\leq 1-\frac{1}{m}$, we can bound $\pi_\text{OP}(\delta)$ via
\begin{align*}
\pi_\text{OP}(\delta) &= \inf_{P_{Z|Y}} I(X;Z) \quad \text{s.t. } \Pr(X \neq Z) \leq \delta' \\
&\geq \inf_{P_{Z|X}} I(X;Z) \quad \text{s.t. } \Pr(X \neq Z) \leq \delta'\\
&=g_\text{OP}(\delta),
\end{align*}
where the inequality is due to the removal of the Markov chain constraint and the final equality follows from Lemma~\ref{lemma_fano1}.
The case when $p > 1-\frac{1}{m}$ follows similarly, except now the term $\left(1-\frac{pm}{m-1}\right) < 0$, hence the constraint $\Pr(Y \neq Z) \leq \delta$ is equivalent to $\Pr(X \neq Z) \geq \delta'$, and $\delta' > 1-\frac{1}{m}$.
Thus, for $p > 1-\frac{1}{m}$, we can bound $\pi_\text{OP}(\delta)$ via
\begin{align*}
\pi_\text{OP}(\delta) &= \inf_{P_{Z|Y}} I(X;Z) \quad \text{s.t. } \Pr(X \neq Z) \geq \delta' \\
&\geq \inf_{P_{Z|X}} I(X;Z) \quad \text{s.t. } \Pr(X \neq Z) \geq \delta'\\
&=g_\text{OP}(\delta),
\end{align*}
where the inequality is due to the removal of the Markov chain constraint and the final equality follows from Lemma~\ref{fano_lower_bound2}.

We next show that $\pi_\text{OP}(\delta) \leq g_\text{OP}(\delta)$, via the mechanism given by $Z := Y + N \mod m$, where $N$ is independent of $(X,Y)$, and distributed according to~\eqref{eqn:OP_opt_mech}.
This mechanism is feasible since $\Pr(Y \neq Z) = t := \min(\delta, 1-\frac{1}{m}) \leq \delta$.
Hence, we have $\pi_\text{OP}(\delta) \leq I(X;Z)$.
For all $x \neq z$, we have that
\begin{align*}
&P_{X,Z}(x,z) = \sum_y P_{Z|Y}(z|y) P_{X,Y}(x,y) \\
&\quad= P_{Z|Y}(z|z) P_{X,Y}(x,z) + \sum_{y \neq z} P_{Z|Y}(z|y) P_{X,Y}(x,y) \\
&\quad= (1-t) P_{X,Y}(x,z) \\
&\quad\quad+ \frac{t}{(m-1)} \left( P_{X,Y}(x,x) + \sum_{y\notin \{x,z\}}P_{X,Y}(x,y) \right) \\
&\quad= \frac{(1-t)p}{(m-1)m} + \frac{t}{(m-1)} \left( \frac{(1-p)}{m} + \frac{p(m-2)}{(m-1)m} \right) \\
&\quad= \frac{p + t \left(1-\frac{pm}{m-1}\right)}{(m-1)m} \\
&\quad=: \frac{\delta''}{(m-1)m} \\
\end{align*}
which shows that $(X,Z) \sim SP(m, \delta'')$.
Thus, by Lemma~\ref{lemma:sp_mut_computation}, $I(X;Z)=r_m(\delta'')$.
Noting that $r_m(1-\frac{1}{m}) = 0$, we have $r_m(\delta'') = g_\text{OP}(\delta)$.
Hence, $\pi_\text{OP}(\delta) \leq g_\text{OP}(\delta)$.

%% file: xavailable.tex
\section{Proof of Theorem~\ref{thm:x_available}}\label{app:proof_x_available}

For convenience, we define 
\begin{align*}
g_\text{INF}(\delta):=
\begin{cases}
r_m(t), & \text{if }
\delta < 1-\frac{1}{m} \text{ and } p \notin(\delta,h), \\
\infty, & \text{if } 
\delta < 1-\frac{1}{m} \text{ and } p \in(\delta,h), \\
0, & \text{if } \delta \geq 1-\frac{1}{m},
\end{cases}
\end{align*}
which is equal to the right-hand side of~\eqref{eqn:SP_INF_region}, where $h := (m-1)(1-\delta)$ and
\[
t:= \frac{\delta-p}{1-\frac{pm}{m-1}}.
\]
Since, for $\delta \geq 1 - \frac{1}{m}$, we immediately have $g_\text{INF}(\delta) = \pi_\text{INF}(\delta) = 0$, we will assume that $\delta < 1 - \frac{1}{m}$ for the rest of this proof.
Note that with this assumption, we have $h > 1 - \frac{1}{m}$.

Since $Y \to X \to Z$ forms a Markov chain for any inference mechanism, we have from Lemma~\ref{lemma:error_relation_markov} that
\[
\Pr(Y \neq Z) = p + \Pr(X \neq Z) \left(1-\frac{pm}{m-1}\right).
\]
Note that if $p = 1-\frac{1}{m}$, then $\Pr(Y \neq Z) = 1-\frac{1}{m} > \delta$, and the optimization is infeasible, hence $g_\text{INF}(\delta) = \pi_\text{INF}(\delta) = \infty$.
Thus, we will consider the two remaining cases: \textbf{(i)} $p < 1-\frac{1}{m}$ and \textbf{(ii)} $p > 1-\frac{1}{m}$. \\

\textbf{Case 1}: $p < 1-\frac{1}{m}$\\

In this case, we have that the constraint $\Pr(Y\neq Z)\leq \delta$ is equivalent to
\[
\Pr(X \neq Z) \leq \frac{\delta-p}{1-\frac{mp}{m-1}} =: t,
\]
due to Lemma~\ref{lemma:error_relation_markov}.
For $p > \delta$, the optimization problem is infeasible since $t < 0$, and hence $g_\text{INF}(\delta) = \pi_\text{INF}(\delta) = \infty$.
Otherwise, for $p \leq \delta$, we have that $0 \leq t \leq 1-\frac{1}{m}$, and by Lemma~\ref{lemma_fano1}, we have that $g_\text{INF}(\delta) = \pi_\text{INF}(\delta) = r_m(t)$.

\textbf{Case 2}: $p > 1-\frac{1}{m}$\\

In this case, we have that the constraint $\Pr(Y\neq Z)\leq \delta$ is equivalent to
\[
\Pr(X \neq Z) \geq \frac{\delta-p}{1-\frac{pm}{m-1}} =: t,
\]
due to Lemma~\ref{lemma:error_relation_markov} and since the denominator is negative.
For $p < h$, the optimization problem is infeasible since $t > 1$, and hence $g_\text{INF}(\delta) = \pi_\text{INF}(\delta) = \infty$.
Otherwise, for $p \geq h$, we have that $1-\frac{1}{m} \leq t \leq 1$, and by Lemma~\ref{fano_lower_bound2}, we have that $g_\text{INF}(\delta) = \pi_\text{INF}(\delta) = r_m(t)$.